\documentclass[a4paper,hidelinks,11pt]{article}
\usepackage[english]{babel}
\usepackage[utf8]{inputenc}
\usepackage{graphicx}
\usepackage{colortbl,booktabs}
\usepackage{bm}
\usepackage{hyperref}
\usepackage{epsfig}
\usepackage{amsmath,amsthm,mathtools}
\usepackage{amsfonts}
\usepackage{here,bbm,url}
\usepackage{amssymb, longtable}
\usepackage{multirow}
\usepackage{lmodern}

\newtheorem{theorem}{Theorem}
\newtheorem{proposition}{Proposition}
\newtheorem{remark}{Remark}

\theoremstyle{thmstylethree}%
\newtheorem{definition}{Definition}%

\newcommand{\xn}{{\bf x}}

\newcommand{\Vn}{{\bf V}}
\newcommand{\vn}{{\bf v}}

\newcommand{\betn}{{\mbox{\boldmath $\beta$}}}

\newcommand{\mun}{{\mbox{\boldmath $\mu$}}}

\newcommand{\tetn}{{\mbox{\boldmath $\theta$}}}

\newcommand{\taun}{{\mbox{\boldmath $\tau$}}}

\begin{document}

\title{New results and regression model for the exponentiated odd log-logistic Weibull family of distributions with applications}

\date{}
\author{
Gabriela M. Rodrigues\\
Departamento de Ci\^encias Exatas\\
Universidade de S{\~a}o Paulo, Piracicaba, SP, Brazil\\
e-mail:\url{gabrielar@usp.br}
\\[0,15cm]
Roberto Vila\\
Departamento de Estat\'{\i}stica,\\
Universidade de Bras\'{\i}lia, DF, Brazil \\
e-mail:\url{rovig161@gmail.com}
\\[0,15cm]
Edwin M. M. Ortega\\
Departamento de Ci\^encias Exatas,\\
Universidade de S{\~a}o Paulo,
Piracicaba, SP, Brazil\\
e-mail:\url{edwin@usp.br}
\\[0,15cm]
Gauss M. Cordeiro\\
Departamento de Estatística,\\
Universidade de Federal de Pernambuco,
Recife, PB, Brazil\\
e-mail:\url{gausscordeiro@gmail.com}
\\[0,15cm]
Victor Serra\\
Departamento de Estat\'{\i}stica,\\
Universidade de Bras\'{\i}lia, DF, Brazil \\
e-mail:\url{victorserra92@gmail.com}
}

\maketitle
\vspace*{0.5cm}

\begin{abstract}
We obtain new mathematical properties of the exponentiated odd log-logistic family of distributions, and of its special case named the exponentiated odd log-logistic Weibull, and its log transformed. A new location and scale regression model is constructed, and some simulations are carried out to verify the behavior of the maximum likelihood estimators, and of the modified deviance-based residuals. The methodology is applied to the Japanese-Brazilian emigration data.
\end{abstract}

\section{Introduction}

In the survival data analysis literature, it is convenient to consider more flexible distributions to capture a wide variety of symmetric, asymmetric
and bimodal behaviors with non-monotonic failure rate function, including as special cases classic distributions, and produce more robust estimates. At present, proposing new distributions to model survival data with non-monotonic failure rate functions is a very important research line in the area of survival analysis. Thus, we initially present new findings for the {\it exponentiated odd log-logistic} (EOLL-G) family which can be employed in several applications to real data. Further, we study two new distributions, called the exponentiated odd log-logistic Weibull (EOLLW) and log exponentiated odd log-logistic Weibull (LEOLLW), and construct a location-scale regression based on the last distribution.

Section \ref{sec:prop_gerais} provides new structural properties of the EOLL-G family.
Sections \ref{sec:prop_eollw} and \ref{sec:prop_log_eollw} define the EOLLW and
LEOLLW distributions and obtain some of their properties. Section \ref{sec:reg} constructs a LEOLLW regression model in location-scale form, reports the maximum likelihood estimates (MLEs), and provides simulations to investigate the accuracy of
the estimates. Section \ref{sec:residuos} define news deviance residuals to assess departures for the propose regression. A real data set is analyzed in Section \ref{sec:application} to show the utility of the new models. Some conclusions are offered in Section \ref{sec:conclusao}.

\section{The EOLL-G density}\label{sec:prop_gerais}

Let $G(x)$ be any baseline cumulative distribution function (cdf) with a parameter vector $\taun$. Alizadeh {\it et al.} (2020) defined the probability density
function (pdf) of the EOLL-G family (for  $x \in \mathbf{R}$) by
\begin{eqnarray}\label{pdf_g}
f(x)=f(x;a,b,\taun)=
\frac{a\,b\, g(x)\,G^{a\,b-1}(x)[1-G(x)]^{a-1} }
 {\left\{G^a(x)+[1-G(x)]^a\right\}^{b+1}},
\end{eqnarray}
where $g(x)=d G(x)/dx$, and $a > 0$ and $b > 0$ are extra shape parameters.

Henceforth, $X \sim \mbox{EOLL-G}(a,b,\taun)$
denotes  a random variable with  pdf \eqref{pdf_g}.
The EOLL-G family becomes the OLL-G family when $b=1$ (Gleaton and Lynch, 2006).
If $a=1$, it is the exponentiated (Exp-G) class (Mudholkar {\it et al.}, 1996).
For $a = b = 1$, Equation (\ref{pdf_g}) leads to the baseline $G(x)$.

If $U \sim U(0,1)$, then
\begin{eqnarray}\label{qf}
Q_{G}\left\{\frac{u^{1/(a\,b)}}{u^{1/(a\,b)}
+(1-u^{{1}/{b}})^{{1}/{a}}}\right\}\sim\mbox{EOLL-G}(a,b,\taun),\,\,
0<u<1,
\end{eqnarray}
where $Q_{G}(u)=G^{-1}(u)$ is the quantile function (qf) of the baseline G model.
In general any distribution in Equation (\ref{pdf_g}) can be simulated from (\ref{qf}) by inverting the parent cdf.

Some EOLL-G properties were addressed by Alizadeh {\it et al.} (2020). We find
new ones below.

\subsection{Properties}
Some EOLL-G properties follow directly by routine methods in calculus.
\begin{itemize}
\item[(P1)] Equation (\ref{pdf_g}) gives $\lim_{x\to\infty} f(x)=0$. If $L=\lim_{x\to 0^+} g(x)$, then
\begin{align*}
\lim_{x\to 0^+} f(x)
=
\begin{cases}
\infty, & \text{if} \ L>0, ab<1;
\\[0,1cm]
ab L, & \text{if} \ L\geqslant 0, ab=1;
\\[0,1cm]
0, & \text{if} \ L=0, ab\geqslant 1.		
\end{cases}
\end{align*}
In general, depending on the choice of the parent, $\lim_{x\to\infty} f(x)=0$ and $\lim_{x\to 0^+} f(x)=ab \lim_{x\to 0^+} [g(x)\,G^{ab-1}(x)]$.
\item[(P2)] For the hazard rate function (hrf) of $X$, say $h(x)$, it follows: $\lim_{x\to 0^+}h(x)=\lim_{x\to 0^+}f(x)$ and
\begin{align*}
\lim_{x\to\infty}h(x)=
ab \lim_{x\to \infty} \left[{g(x)\,G^{ab-1}(x)\over 1-F(x)}\right],
\end{align*}
where
\begin{align*}
h(x)
=
\frac{ab\, h_G(x)\,G^{ab-1}(x)[1-G(x)]^{a} }
{\left\{G^a(x)+[1-G(x)]^a\right\} \bigl[\left\{G^a(x)+[1-G(x)]^a\right\}^b-G^{ab}(x)\bigr]},
\end{align*}
where $h_G(x)=g(x)/[1-G(x)]$ is the baseline hrf.
\item[(P3)]
A straightforward derivative computation leads to
\begin{align}\label{f-derivative}
f'(x)
=
f(x)
\left\{
{T''(x)}-{(a+1)T^a(x)+(1-ab)\over T(x)[1+T^a(x)]}\, [T'(x)]^2
\right\},	
\end{align}
where $T(x)=G(x)/[1-G(x)]$ and $T'(x)=g(x)/[1-G(x)]^2$.
Then, the critical points of the pdf of $X$ are the roots of:
\begin{align*}
{T''(x)\over [T'(x)]^2}={(a+1)T^a(x)+(1-ab)\over T(x)[1+T^a(x)]},
\end{align*}
with
$
T''(x)
=
g'(x)/[1-G(x)]^2+{2g^2(x)/[1-G(x)]^3}
$.
\item[(P4)]
If there is a function $r(x)$ such that the differential equation
$g'(x)=-g(x)[h_G(x)+ r(x)]$ holds, Equation \eqref{f-derivative} can be written as
\begin{align}\label{f-derivative-1}
f'(x)
=
f(x)h_G(x)
\left[
1-\left\{1+{a[T^a(x)-b]\over T^a(x)+1} \right\}{1\over G(x)}-{r(x)\over h_G(x)}
\right],	
\end{align}
since $h_G(x)>0$.
Hence, from \eqref{f-derivative-1} and (P3), the critical points of the pdf of $X$ are the roots of:
\begin{align}\label{primitive-modes}
\biggl[1-{r(x)\over h_G(x)}\biggr]\,G(x)
=
1+{a[T^a(x)-b]\over T^a(x)+1}.
\end{align}
\item[(P5)]
Let $X \sim \mbox{EOLL-G}(a,b,\taun)$.
If $D$ has the Dagum distribution Type I (Dagum 1975), say $D\sim {\rm DAGUM}(a,1,b)$, then
\begin{align*}
F(x)=\mathbbm{P}[D\leqslant T(x)]
=
\mathbbm{P}\biggl[G^{-1}\biggl({D\over 1+D}\biggr)\leqslant x\biggr],
\end{align*}
where  $T(x)=G(x)/[1-G(x)]$. Consequently, the stochastic representation for $X$ holds
\begin{align*}
X=G^{-1}\biggl({D\over 1+D}\biggr).
\end{align*}

\item[(P6)]
Let $B=1/D$. It is well-known that $D\sim {\rm DAGUM}(a,1,b)$ $\Longleftrightarrow$ $B\sim {\rm BURR}(a,1,b)$, where $B$ has the Burr Type XII distribution (Burr 1942). Hence, by (P5),
\begin{align*}
F(x)=
\mathbbm{P}[B\geqslant S(x)]
=
\mathbbm{P}\biggl[G^{-1}\biggl({1\over 1+B};\taun\biggr)\geqslant x\biggr],
\end{align*}
where $S(x)=1/T(x)$.
\end{itemize}
	
\section{The EOLLW distribution}\label{sec:prop_eollw}
Consider the parent Weibull cdf
$G(x)=
1-\exp\left\{-\left(x/\lambda\right)^\alpha\right\}$, where $\lambda>0$ is a scale, and $\alpha>0$ is a shape. The EOLLW pdf is determimed from (\ref{pdf_g}) (for $x >0$) as
\begin{eqnarray}\label{pdf_eollw}
f(x)=
\frac{ab\alpha\,x^{\alpha-1}\exp\left\{-a\left(\frac{x}{\lambda}\right)^\alpha\right\} \left[1-\exp\left\{-\left(\frac{x}{\lambda}\right)^\alpha\right\}\right]^{ab-1} }
{\lambda^{\alpha}\left\{\left[1-\exp\left\{-\left(\frac{x}{\lambda}\right)^\alpha\right\}\right]^a+\exp\left[-a\left(\frac{x}{\lambda}\right)^\alpha\right]\right\}^{b+1}}.
\end{eqnarray}
Equation  (\ref{pdf_eollw}) yields $\lim_{x\to\infty}f(x)=0$. Further,
for $\alpha>0$,
\begin{align}\label{lim-exp-weibull}
	\lim_{x\to 0^+}
	f(x)
	=
	\begin{cases}
	\infty, & ab<1/\alpha,
	\\[0,1cm]
	1/\lambda^\alpha, & ab=1/\alpha,
	\\[0,1cm]
	0, & ab>1/\alpha.
	\end{cases}
\end{align}

The EOLLW distribution is very flexible due to different forms of its pdf and hrf;
see Figures \ref{eollwpdf} and \ref{eollwhrf} and Sections \ref{Modality-EOLLW}, \ref{Shape-EOLLW} and \ref{Tail-EOLLW} (for theoretical results).

\begin{figure}[h]
	\centering	
	(a)\hspace{4cm} (b) \hspace{4cm} (c)
	\includegraphics[width=1\textwidth]{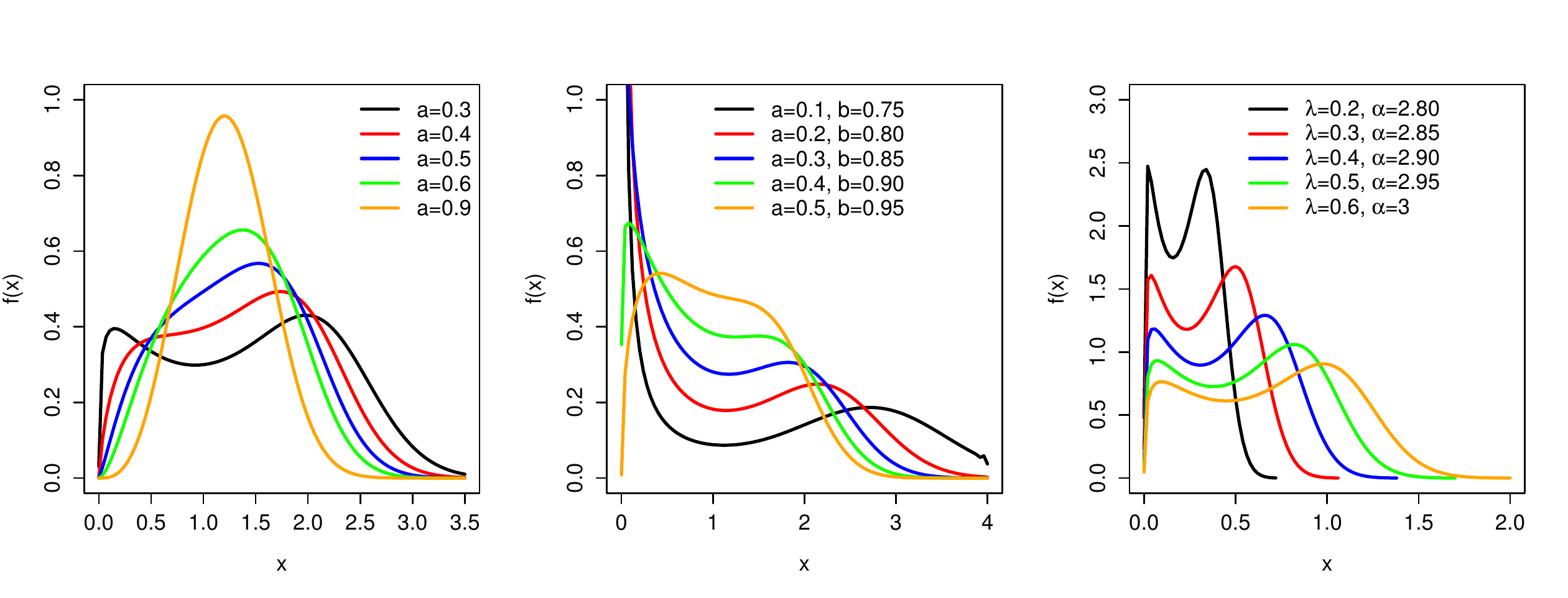}
	\caption{Plots of the EOLLW pdf. (a) For $b=1.5$, $\lambda=1.2$ and $\alpha=2.9$. (b) For $\lambda=1.2$ and $\alpha=3$. (c) For $\alpha$ and $a=0.3$ and $b=1.5$.}
	\label{eollwpdf}
	\end{figure}

\begin{figure}[h]
	\centering	
	(a)\hspace{4cm} (b) \hspace{4cm} (c)	
	\includegraphics[width=1\textwidth]{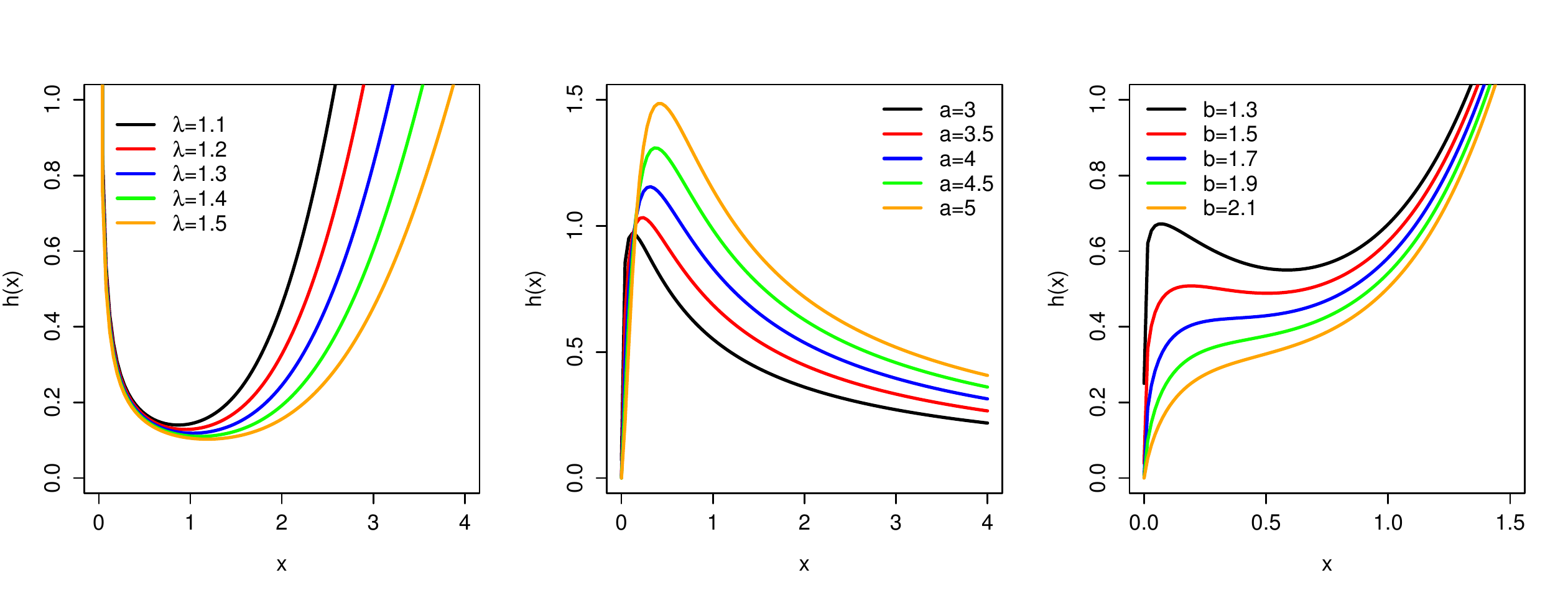}
	\caption{Plots of the EOLLW hrf. (a) For $\lambda$ and $a=0.1$, $b=2$ and $\alpha=3$. (b) For $b=2$, $\lambda=2$ and $\alpha=0.2$. (c) For $a=0.3$, $\lambda=1$ and $\alpha=2.9$.}
	\label{eollwhrf}
\end{figure}

\subsection{Modality of the EOLLW density}
\label{Modality-EOLLW}
Since $G(x)=1-\exp\left\{-\left(x/\lambda\right)^\alpha\right\}$, $T$ defined in (P3) is written as
\begin{align}\label{T-EOLLW}
T(x)
=
T(x;\alpha,\lambda)
=
\exp\Biggl\{\biggl({{x}\over{\lambda}}\biggr)^\alpha\Biggr\}-1.
\end{align}

Since  $T(x;\alpha,\lambda)=T(x;\alpha,k\lambda)$, $k > 0$, and since $F(x)=\mathbbm{P}[D\leqslant T(x)]$, with $D\sim {\rm DAGUM}(a,1,b)$ [see Property (P5)], the next result follows.
\begin{proposition}
	If $X \sim{\rm EOLLW}(a,b,\taun)$ with $\taun=(\alpha,\lambda)^T$, then $kX \sim{\rm EOLLW}(a,b,\widetilde{\taun})$ with $\widetilde{\taun}=(\alpha,k\lambda)^T$.
\end{proposition}

A critical point of the EOLLW density by property (P4) is a positive root
of the nonlinear Equation \eqref{primitive-modes}.
A straightforward computation gives $g'(x)=-g(x) \,[h_G(x)+r(x)]$, where
$h_G(x)=(\alpha/\lambda)(x/\lambda)^{\alpha-1}$ is the hrf  of the parent $G$,
and $r(x)=(1-\alpha)/x$.
By inserting $T(x)$ in \eqref{primitive-modes}, a critical point of the EOLLW density is a root of:
\begin{align}\label{modes-1}
\mathbbm{G}_0(\omega)=\mathbbm{H}_0(\omega), 	\quad  \omega=\biggl(\frac{x}{\lambda}\biggr)^\alpha,
\end{align}
where
\begin{align}\label{G-H-def}
\mathbbm{G}_0(\omega)
=
\biggl(
1-{1-\alpha\over \alpha \omega}
\biggr)
[1-\exp(-\omega)]
\quad \text{and} \quad
\mathbbm{H}_0(\omega)
=
1+
{a\left\{\left[\exp\left(\omega\right)-1\right]^a-b\right\}\over \left[\exp\left(\omega\right)-1\right]^a+1}.
\end{align}
	
The first two derivatives of $\mathbbm{G}_0$ with respect to $\omega$ are
\begin{align*}
\mathbbm{G}'_0(\omega)
=
{ \exp(-\omega)[\alpha\omega^2+(\alpha-1)\omega+\alpha-1]-(\alpha-1)
\over \alpha \omega^2}
\end{align*}
and
\begin{align*}
\mathbbm{G}''_0(\omega)
=
-{ \exp(-\omega)\left\{\alpha\omega^3+(1-\alpha)[-\omega^2-2\omega+2\exp(\omega)-2]\right\}
\over \alpha \omega^3}.
\end{align*}
For $\omega>0$, $(1-\alpha)[-\omega^2-2\omega+2\exp(\omega)-2]>0$ if $\alpha\leqslant 1$, and then $\mathbbm{G}''_0(\omega)<0$. Moreover, $\mathbbm{G}'_0(\omega)> 0$
for  $\omega>0$, if and only if $[\alpha/(\alpha-1)]\omega^2+\omega+1<\exp(\omega)$ (for $\alpha\neq 1$), which is true for all $\alpha<1$. For $\alpha=1$, $\mathbbm{G}'_0(\omega)> 0$ when $\omega>0$. Further, for $\alpha>1$, $\mathbbm{G}'_0(\omega)= 0$ if and only if $[\alpha/(\alpha-1)]\omega^2+\omega+1=\exp(\omega)$. Since $\alpha>1$, it is natural to expect both functions $[\alpha/(\alpha-1)]\omega^2+\omega+1$ and $\exp(\omega)$ to intersect at a single point. Briefly, we have
\begin{align}\label{shape-G}
\mathbbm{G}_0 \ \text{is strictly concave and increasing for} \ 0<\alpha\leqslant 1, \ \text{and is unimodal for} \ \alpha>1.
\end{align}
Further,
\begin{align*}
\lim_{\omega\to 0^+}\mathbbm{G}_0(\omega) ={\alpha-1\over\alpha}
\quad \text{and} \quad
\lim_{\omega\to \infty}\mathbbm{G}_0(\omega) =1.
\end{align*}
	
On the other hand, the first-order derivative de $\mathbbm{H}_0$ with respect to $\omega$ holds
\begin{align}\label{derivative-H}
\mathbbm{H}'_0(\omega)
=
\frac{a^2 (1 + b) \exp(\omega) [\exp(\omega)-1]^{a-1}}{
\{[\exp(\omega)-1]^a+1\}^2}>0
\quad \text{for all} \ \omega>0.
\end{align}
Setting $z=\exp(\omega)-1$, the second-order derivative of $\mathbbm{H}_0$
with respect to $\omega$ is
\begin{align*}
\mathbbm{H}''_0(\omega)
=
-\frac{
a^2 (b + 1) (z+1) z^{a - 2}
\bigl[
az^{a+1}+(1+a)z^a-az+1-a
\bigr]
}{(z^a + 1)^3}<0 \quad \text{for all} \ a\in \mathcal{A},
\end{align*}
where
$
\mathcal{A}
=
\left\{
a>0:
az^{a+1}+(1+a)z^a-az+1-a>0, \ \forall z>0
\right\}.
$
The set $\mathcal{A}$ is non-empty because $1\in \mathcal{A}$. Further,
it can be proven that $\mathcal{A}=(0,1]$. Thus,
\begin{align}\label{shape-H}
\mathbbm{H}_0 \ \text{is strictly concave and increasing for} \ 0<a\leqslant 1,
\end{align}
with
\begin{align*}
\lim_{\omega\to 0^+}\mathbbm{H}_0(\omega)=1-ab
\quad \text{and} \quad
\lim_{\omega\to \infty}\mathbbm{H}_0(\omega)=1+a.
\end{align*}

\begin{proposition}\label{prop-one-root}
Equation \eqref{modes-1} has at least one root on $(0,\infty)$ when $ab>1/\alpha$.
\end{proposition}
\begin{proof}
Since $\lim_{\omega\to 0^+}\{\mathbbm{G}_0(\omega)-\mathbbm{H}_0(\omega)\}=ab-(1/\alpha)$ and $\lim_{\omega\to \infty}\{\mathbbm{G}_0(\omega)-\mathbbm{H}_0(\omega)\}=-a$, the proof follows
by using the intermediate value theorem.
\end{proof}
	
The previous proposition guarantees the existence of a critical point of the
EOLLW density if $ab>1/\alpha$. The following result, under certain restrictions
on the parameters, shows that this critical point is unique.

\begin{theorem}\label{theo-shape-0}
If $\alpha=1$ and $a>0$ is an integer, the shape of the {\rm EOLLW} density is
\begin{enumerate}
\item
decreasing or decreasing-increasing-decreasing if $ab<1$;
\item
unimodal if $ab\geqslant 1$.
\end{enumerate}
\end{theorem}
\begin{proof}
For $\alpha=1$, Equation \eqref{modes-1} becomes
\begin{align*}
p(z)=az^{a+1}+(1+a)z^a-abz+(1-ab)=0, \quad z=\exp(\omega)-1,
\end{align*}
where $a$ is an integer. The number of zeros of $p(z)$ determines the number
of the critical points of the EOLLW distribution.
		
Let $ab<1$. By Descartes' rule of signs (Griffiths, 1947; and Xue, 2012), the polynomial $p(z)$ has two sign changes (the sequence signs is $+, +, -, +$),
meaning that $p(z)=0$ has two or zero positive roots. First, assume that $p(z)=0$ has two positive roots, say $z_1$ and $z_2$. Then, the EOLLW density has two
critical points $x_1 = \lambda\log(1+z_1)$ and $x_2 = \lambda\log(1+z_2)$. Since $\lim_{x\to 0^+} f(x)=\infty$ and $\lim_{x\to\infty} f(x)=0$
when $ab<1$, see \eqref{lim-exp-weibull}, it follows that the EOLLW
pdf is decreasing-increasing-decreasing.
Second, if $p(z)=0$ has zero positive roots, then the EOLLW pdf has no critical
point. Since $f(x)$ explodes at the origin, the EOLLW density is decreasing.
		
On the other hand, let $ab>1$. Again, by Descartes' rule of signs,
the polynomial $p(z)$ has one sign change (the sequence signs is $+, +, -, -$). This means that $p(z)=0$ has a unique positive root, and then the EOLLW pdf has one critical point. By  \eqref{lim-exp-weibull}, $\lim_{x\to 0^+} f(x)=\lim_{x\to\infty} f(x)=0$
when $ab>1$, and then the unimodality of the EOLLW pdf follows.
		
The proof for the case $ab=1$ follows by combining the limit in \eqref{lim-exp-weibull} with steps analogous to the proof for $ab>1$. Therefore, this is omitted.
\end{proof}
	
It is an arduous task to find (or provide optimal above bounds for the number of)
roots of general nonlinear equations. Numerical methods are suitable for this
purpose. From the facts that $\mathbbm{G}_0$ and $\mathbbm{H}_0$ have the
shapes in \eqref{shape-G} and \eqref{shape-H}, respectively, and that $\mathbbm{G}_0-\mathbbm{H}_0$ is not a periodic function,
it is plausible that (depending on the parameters chosen) that
Equation \eqref{modes-1} has at most three roots, but we do not have a
proof. Using the Rouche's theorem related to the number of roots in discs
centered at zero, perhaps this can be useful to deal with this question.
In order to establishing the shape of the EOLLW distribution, we suppose
that $\mathbbm{G}_0-\mathbbm{H}_0$ has at most three zeros, i.e.,
we have the following scenarios:
\begin{align*}
\begin{array}{lllll}
&\text{(i)}
& \text{Plots of} \ \mathbbm{G}_0 \ \text{and} \ \mathbbm{H}_0 \ \text{do not have a point of intersection;}
\\[0.15cm]
&\text{(ii)}
& \text{Plots of} \ \mathbbm{G}_0 \ \text{and} \ \mathbbm{H}_0 \ \text{have a single common point;}
\\[0.15cm]
&\text{(iii)}
& \text{Plots of} \ \mathbbm{G}_0 \ \text{and} \ \mathbbm{H}_0 \ \text{have two common  points; and}
\\[0.15cm]
&\text{(iv)}
& \text{Plots of} \ \mathbbm{G}_0 \ \text{and} \ \mathbbm{H}_0 \ \text{have three common points.}
\end{array}
\end{align*}
		
\begin{theorem}\label{theo-shape-1}
If $X\sim{\rm EOLLW}(a,b,\taun)$ and
$0<a \leqslant 1$, the shape of the pdf of $X$ is
\begin{enumerate}
\item decreasing or decreasing-increasing-decreasing when $ab<1/\alpha$;
\item decreasing or uni/bimodal or decreasing-increasing-decreasing when $ab=1/\alpha$.
\end{enumerate}
\end{theorem}
\begin{proof}
Equation \eqref{lim-exp-weibull} gives $\lim_{x\to 0^+} f(x)=\infty$ and $\lim_{x\to\infty} f(x)=0$ when $ab<1/\alpha$, which implies that:
for scenario (i), the {\rm EOLLW} pdf has no critical point,
and so this one is decreasing;
for scenario (ii), the {\rm EOLLW} pdf has exactly one critical point,
but this is a contradiction with the fact that the pdf explodes at the origin
and disappears at infinity, and then this case cannot occur;
for scenario (iii) the  {\rm EOLLW} pdf has two critical points,
and then this one is decreasing-increasing-decreasing;
and for scenario (iv) the {\rm EOLLW} pdf has three critical points
which leads to a contradiction by using the same argument as in scenario (ii).
This proves the first item.

By Equation \eqref{lim-exp-weibull},
$\lim_{x\to 0^+}
f(x)=1/\lambda^\alpha
$
and
$\lim_{x\to\infty} f(x)=0$ when $ab= 1/\alpha$. This ensures that,
for scenario (i) the  {\rm EOLLW} pdf has no critical point, and so this one is decreasing; for (ii) the {\rm EOLLW} pdf has exactly one critical point, and then
it is unimodal with mode greater than $1/\lambda^\alpha$;
for (iii) the {\rm EOLLW} pdf has two critical points, and then this
one is decreasing-increasing-decreasing; and for (iv) the {\rm EOLLW} pdf has three critical points, and then this one is bimodal. This proves the second item.
\end{proof}	
	
\begin{theorem}\label{theo-shape-2}
If $X \sim{\rm EOLLW}(a,b,\taun)$ and
$0<a\leqslant 1$, then the shape of the pdf of $X$ is uni-or bimodal
when $a b>1/\alpha$.
\end{theorem}
\begin{proof}
Equation \eqref{lim-exp-weibull} gives $\lim_{x\to 0^+} f(x)=\lim_{x\to\infty} f(x)=0$ when $ab> 1/\alpha$. Consequently, for scenario (i) the {\rm EOLLW} pdf must be the zero function, which cannot occur; for (ii) the {\rm EOLLW} pdf has exactly one critical point, and then this one is unimodal; for (iii) the {\rm EOLLW} pdf has two critical points, but this is a contradiction with the definiton of a pdf, and therefore this case cannot occur; and for (iv) the {\rm EOLLW} pdf has three critical points, and then this one is bimodal. This proves the theorem.
\end{proof}
	
\begin{remark}
Notice that, by Theorems \ref{theo-shape-0}, \ref{theo-shape-1} and \ref{theo-shape-2}, the shape of the {\rm EOLLW} pdf is independent of the choice of $\lambda$.
\end{remark}
\begin{remark}
Since  $\lim_{x\to 0^+} f(x)=\lim_{x\to\infty} f(x)=0$ when $ab> 1/\alpha$, an alternative proof of Theorem \ref{theo-shape-2} follows immediately by the
application of Proposition \ref{prop-one-root}.
\end{remark}
\begin{remark}	
For some values of $a$, $b$ and $\alpha$, we find the table:
\begin{align*}
\begin{array}{ccccccc}
\toprule
a & b & ab & \alpha & 1/\alpha & {\rm inequality} & {\rm modality} \\
\midrule
0.3 & 1.5 & 0.45 & 2.9 & 0.34 & ab>1/\alpha & {\rm U}$-${\rm or \, B} \\
0.4 & 1.5 & 0.60 & 2.9 & 0.34 & ab>1/\alpha &{\rm U}$-${\rm or \, B} \\
0.5 & 1.5 & 0.75 & 2.9 & 0.34 & ab>1/\alpha &{\rm U}$-${\rm or \, B} \\
0.6 & 1.5 & 0.90 & 2.9 & 0.34 & ab>1/\alpha &{\rm U}$-${\rm or \, B} \\
0.9 & 1.5 & 1.35 & 2.9 & 0.34 & ab>1/\alpha &{\rm U}$-${\rm or \, B} \\
\bottomrule
\end{array}
\end{align*}
By applying Theorem \ref{theo-shape-2}, the {\rm EOLLW} pdf is uni-or
bimodal (U-or B), which is supported with the density plot in Figure \ref{eollwpdf} (a).
		
For some values of $a$, $b$ and $\alpha$, we find:
\begin{align*}
\begin{array}{ccccccc}
\toprule
a & b & ab & \alpha & 1/\alpha & {\rm inequality} & {\rm modality} \\
\midrule
0.1 & 0.75 & 0.08 & 3 & 0.33 & ab<1/\alpha & {\rm D}$-${\rm or \ DID} \\
0.2 & 0.80 & 0.16 & 3 & 0.33 & ab<1/\alpha &{\rm D}$-${\rm or \, DID} \\
0.3 & 0.85 & 0.25 & 3 & 0.33 & ab<1/\alpha &{\rm D}$-${\rm or \, DID} \\
0.4 & 0.90 & 0.36 & 3 & 0.33 & ab>1/\alpha &{\rm U}$-${\rm or \, B} \\
0.5 & 0.95 & 0.48 & 3 & 0.33 & ab>1/\alpha &{\rm U}$-${\rm or \, B} \\
\bottomrule
\end{array}
\end{align*}
By applying Item 1 of Theorem \ref{theo-shape-1} and  Theorem \ref{theo-shape-2},
the {\rm EOLLW} pdf is decreasing or decreasing-increasing-decreasing
(D-or DID) and uni-or bimodal (U-or B), respectively, which is in
agreement with the pdf plot in Figure \ref{eollwpdf} (b).	
		
For some values of $a$, $b$ and $\alpha$, it follows:
\begin{align*}
\begin{array}{ccccccc}
\toprule
a & b & ab & \alpha & 1/\alpha & {\rm inequality} & {\rm modality} \\
\midrule
0.3 & 1.5 & 0.45 & 2.80 & 0.36 & ab>1/\alpha & {\rm U}$-${\rm or \, B} \\
0.3 & 1.5 & 0.45 & 2.85 & 0.35 & ab>1/\alpha &{\rm U}$-${\rm or \, B} \\
0.3 & 1.5 & 0.45 & 2.90 & 0.34 & ab>1/\alpha &{\rm U}$-${\rm or \, B} \\
0.3 & 1.5 & 0.45 & 2.95 & 0.34 & ab>1/\alpha &{\rm U}$-${\rm or \, B} \\
0.3 & 1.5 & 0.45 & 3.00 & 0.33 & ab>1/\alpha &{\rm U}\text{-}{\rm or \, B} \\
\bottomrule
\end{array}
\end{align*}
By applying Theorem \ref{theo-shape-2}, the {\rm EOLLW} pdf is uni-or bimodal (U-or B), which is compatible with the pdf plot in Figure \ref{eollwpdf} (c).
\end{remark}

\subsection{Shapes of the EOLLW hrf} \label{Shape-EOLLW}
Let $\eta(x)= -{f'(x)/ f(x)}$, where $f(x)$ is the EOLLW pdf in \eqref{pdf_eollw}.  Following Glaser (1980), we characterize the hrf $h(x)$ of $X \sim{\rm EOLLW}(a,b,\taun)$ through $\eta$.
	
For simplicity, let $\alpha=1$. By \eqref{f-derivative-1}, $\eta$ can be
written as
\begin{align*}
\eta(x)
=
{1\over\lambda}
\left[{ \mathbbm{H}_0(\omega) \over 1-\exp(-\omega) }-1\right],
\quad \omega=\frac{x}{\lambda},
\end{align*}
where $\mathbbm{H}_0(\omega)$ is as in \eqref{G-H-def}. By differentiating $\eta$
with respect to $x$, and using the formula \eqref{derivative-H} for $\mathbbm{H}_0'(\omega)$, we obtain
\begin{align}\label{eta-derivative}
\eta'(x)
=
{(t+1)\over\lambda^2 t^2(t^a+1)^2}\,
{p_{a,b}(t)}, \quad
t=\exp(\omega)-1,
\end{align}
where
\begin{align}\label{polynomial-ab}
p_{a,b}(t)
=
(1+a)t^{2a}+a^2(1+b)t^{a+1}+[a^2(1+b)+a(1-b)+2]t^a+(1-ab).	
\end{align}
So, the number of roots of $p_{a,b}(t)$ determines the number of  critical points
of $\eta$.
	
To state and prove the next result, we define
\begin{align*}
\mathcal{B}
=\{(a,b)\in(0,\infty)\times(0,\infty): p_{a,b}(t)>0, \ \forall t>0\}.
\end{align*}
By choosing $a,b>0$ such that $a^2(1+b)+a(1-b)+2\geqslant 0$
and $ab\leqslant 1$, we guarantee that the set $\mathcal{B}$ is non-empty.
	
\begin{theorem}
Let $X \sim{\rm EOLLW}(a,b,\taun)$ and $\alpha=1$.
\begin{enumerate}
\item
If $(a,b)\in \mathcal{B}$, the hrf of $X$ is increasing.
\item
Suppose $a^2(1+b)+a(1-b)+2\geqslant 0$ and $ab>1$, and $a>0$ integer.
For example, take $a\geqslant 1$ integer and $b>1/a$.
\begin{enumerate}
\item If there exists $0<x^*<a/\lambda $ such that $h'(x^*)= 0$, then the hrf of $X$ has bathtub (BT) shape.
\item If there does not exist $0<x^*<a/\lambda $ such that $h'(x^*)= 0$,
then the hrf of $X$ is increasing.
\end{enumerate}
\item Let $a^2+3a-1>0$ and $ab=1$, and $a>0$ integer. For example, take $0<a<(\sqrt{13}-3)/2$ and  $b>1/a$. Under the condition of Item (a) $[$respectively, Item (b)$]$, the hrf of $X$ has BT shape $[$respectively, is increasing$]$.
\end{enumerate}
\end{theorem}
\begin{proof}
For $(a,b)\in \mathcal{B}$, then $p_{a,b}(t)>0$, $\forall t>0$. Hence, by
Equation \eqref{eta-derivative}, $\eta'(x)>0$ for all $x>0$. So, Item 1
holds (Glaser, 1980).
		
In what follows, we prove the statement in Item 2. Under the conditions imposed in  this one: $a^2(1+b)+a(1-b)+2\geqslant 0$ and $ab>1$, for $a>0$ integer;
by Descartes' rule of signs, the polynomial $p_{a,b}(t)$ in \eqref{polynomial-ab}
has one sign change (the sequence signs is $+, +, +, -$), thus meaning that this polynomial has a single positive root. Then, from Equation \eqref{eta-derivative},
it follows that $\eta'(x)$ has a single positive root, say $x^*$. Since $\lim_{x\to 0^+}\eta(x)=\infty$ and $\lim_{x\to \infty}\eta(x)=a/\lambda$,  we obtain
$\eta'(x) < 0$ for $x<x^*$ $(<a/\lambda)$, $\eta'(x^*)=0$,
and $\eta'(x)>0$ for $x>x^*$. Under the hypothesis (a) [respectively, hypothesis (b)], the hrf of $X$ has BT shape [respectively, is increasing] (Glaser, 1980).
		
The proof of Item 3 follows by using the same steps as in Item 2, so it is omitted.
\end{proof}

\subsection{Tail behavior of EOLLW}\label{Tail-EOLLW}

\begin{definition}\label{def-tail}
	Let $F$ be a continuous univariate distribution on $\mathbf{R}$, and  $\overline{F}(x)=1-F(x)$.
	\begin{enumerate}
		\item
		The distribution $F$ has least light-tail distribution if, for any $t > 0$,
		\begin{align*}
		\lim_{x\to -\infty} {\exp(tx)\over F(x)}=\infty.
		\end{align*}
		\item
		The distribution F has upper light-tail distribution if, for any $t > 0$,
		\begin{align*}
		\lim_{x\to \infty} {\exp(-tx)\over \overline{F}(x)}=\infty.
		\end{align*}
		\item
		The distribution F has least heavy-tail distribution if, for any $t > 0$,
		\begin{align*}
		\lim_{x\to -\infty} {\exp(tx)\over F(x)}=0.
		\end{align*}
		\item
		The distribution F has upper heavy-tail distribution if, for any $t > 0$,
		\begin{align*}
		\lim_{x\to \infty} {\exp(-tx)\over \overline{F}(x)}=0.
		\end{align*}
	\end{enumerate}
\end{definition}

In this subsection, we prove that the EOLLW model has a transition from heavy-tailed distributions to light-tailed (Figure \ref{transition-tail}).
\begin{figure}[H]
	\centering	
	\includegraphics[width=0.43\textwidth]
	{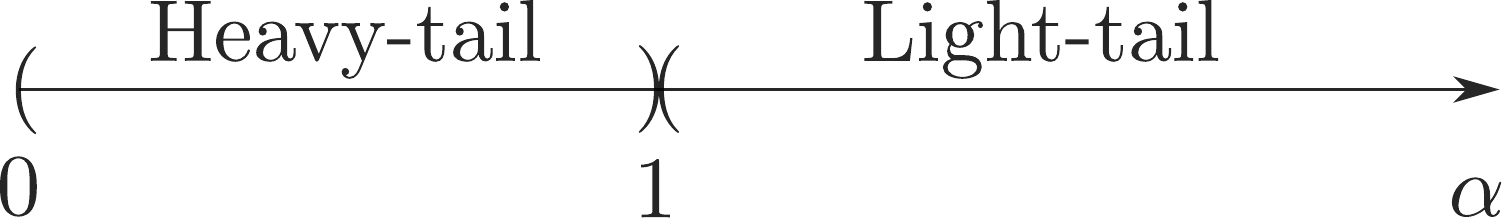}
	\caption{Heavy-tailedness and light-tailedness for the EOLLW model.}
	\label{transition-tail}
\end{figure}

\begin{proposition}
	The shape parameter $\alpha$ governs the tail
	behavior of the {\rm EOLLW} distribution type of the following form:
	\begin{enumerate}
		\item[(a)]
		If $\alpha<1$ then the {\rm EOLLW} has upper heavy-tail distribution.
		\item[(b)]
		If $\alpha=1$ then the limit
		\begin{align}\label{limit-dep-t}
		\lim_{x\to \infty}
		{
			\exp(-tx)
			\over
			\overline{F}(x)
		} \
		\text{depends on} \ t>0.
		\end{align}
		\item[(c)]
		If $\alpha>1$ then the {\rm EOLLW} has upper light-tail distribution.
	\end{enumerate}
\end{proposition}
\begin{proof}
	By Property (P5), $F(x)=\mathbbm{P}[D\leqslant T(x)]\eqqcolon F_D[T(x)]$, where $D\sim {\rm DAGUM}(a,1,b)$ (Dagum distribution Type I) and $T$ is as in \eqref{T-EOLLW}.
	Moreover, it is well-known that
	$F_D[d]=(1+d^{-a})^{-b}$.
	Then, for any $t>0$, we have
	\begin{align}\label{ineq-in}
	\lim_{x\to \infty}
	{
		\exp(-tx)
		\over
		\overline{F}(x)
	}
	=
	\lim_{x\to \infty}
	{
		\exp(-tx)
		\over
		\overline{F}_D[T(x)]
	}
	&=
	\lim_{x\to \infty}
	{
		\exp(-tx)
		\over
		1-[1+T(x)^{-a}]^{-b}
	}
	\eqqcolon
	L.
	\end{align}
	Since $T(x)=\exp[(x/\lambda)^\alpha]-1$, L'Hôpital's rule gives
	\begin{align}\label{hospital}
	L
	&=
	\lim_{x\to \infty}
	{
		t\exp(-tx)
		\over
		ab[1+T(x)^{-a}]^{-b-1} T(x)^{-a-1}T'(x)
	}
	\nonumber
	\\[0,2cm]
	&=
	{\lambda t\over ab\alpha}
	\lim_{x\to \infty}
	{
		[1+T(x)^{-a}]^{b+1} \,
		{\{\exp[({x\over\lambda})^\alpha]-1\}^{a+1} \over ({x\over\lambda})^{\alpha-1}
			\exp[({x\over\lambda})^\alpha+tx]}
	}.
	\end{align}
	But, 	
	\begin{align}\label{limt-three}
	\lim_{x\to \infty}T(x)=\infty
	\quad \text{and} \quad
	\lim_{x\to \infty}
	{
		{\{\exp[({x\over\lambda})^\alpha]-1\}^{a+1} \over ({x\over\lambda})^{\alpha-1}
			\exp[({x\over\lambda})^\alpha+tx]}
	}=
	\begin{cases}
	\infty, & \alpha>1;
	\\
	\infty, & \alpha=1, \, a>\lambda t;
	\\
	0, & \alpha=1, \, a<\lambda t;
	\\
	1, & \alpha=1, \, a=\lambda t;
	\\
	0, & \alpha<1.
	\end{cases}
	\end{align}
	By combining \eqref{limt-three} and \eqref{hospital}  with \eqref{ineq-in}, we obtain that (for any $t>0$)
	\begin{align*}
	\lim_{x\to \infty}
	{
		\exp(-tx)
		\over
		\overline{F}(x)
	}
	=
	\begin{cases}
	\infty, & \alpha>1;
	\\
	0, & \alpha<1;
	\end{cases}
	\end{align*}
	and that, for $\alpha=1$, the limit in \eqref{limit-dep-t} is a function of $t>0$.
	This completes the proof.
\end{proof}


\section{The LEOLLW distribution}\label{sec:prop_log_eollw}
If $X\sim \mbox{EOLLW}(\alpha,\lambda,a,b)$, $Y=\log(X)$ has the LEOLLW density
written in terms of $\mu=\log(\lambda)$ and $\sigma=1/\alpha$ (for $y\in \mathbf{R}$)
\begin{eqnarray}\label{fy}
f(y)=
\frac{ab \exp\left\{\left(\frac{y-\mu}{\sigma}\right)-a\exp\left(\frac{y-\mu}{\sigma}\right)\right\}
\left[1-\exp\left\{-\exp\left(\frac{y-\mu}{\sigma}\right)\right\}\right]^{ab-1}  }
{\sigma\left\{\left[1-\exp\left\{-\exp\left(\frac{y-\mu}{\sigma}\right)\right\}\right]^a+\exp\left[-a\exp\left(\frac{y-\mu}{\sigma}\right)\right]\right\}^{b+1}},
\end{eqnarray}
where $a > 0$, $b>0$, $\sigma > 0$ and $\mu \in \mathbf{R}$.

Henceforth, let $Y \sim \mbox{LEOLLW}(\mu,\sigma,a,b)$ be a
random variable with pdf \eqref{fy}. Plots of the pdf of $Y$ displayed
in Figure \ref{fdp_y} reveal great flexibility of the new density.
\begin{figure}[h]
\centering		
(a)\hspace{4cm} (b) \hspace{4cm} (c)
\includegraphics[width=1\textwidth]{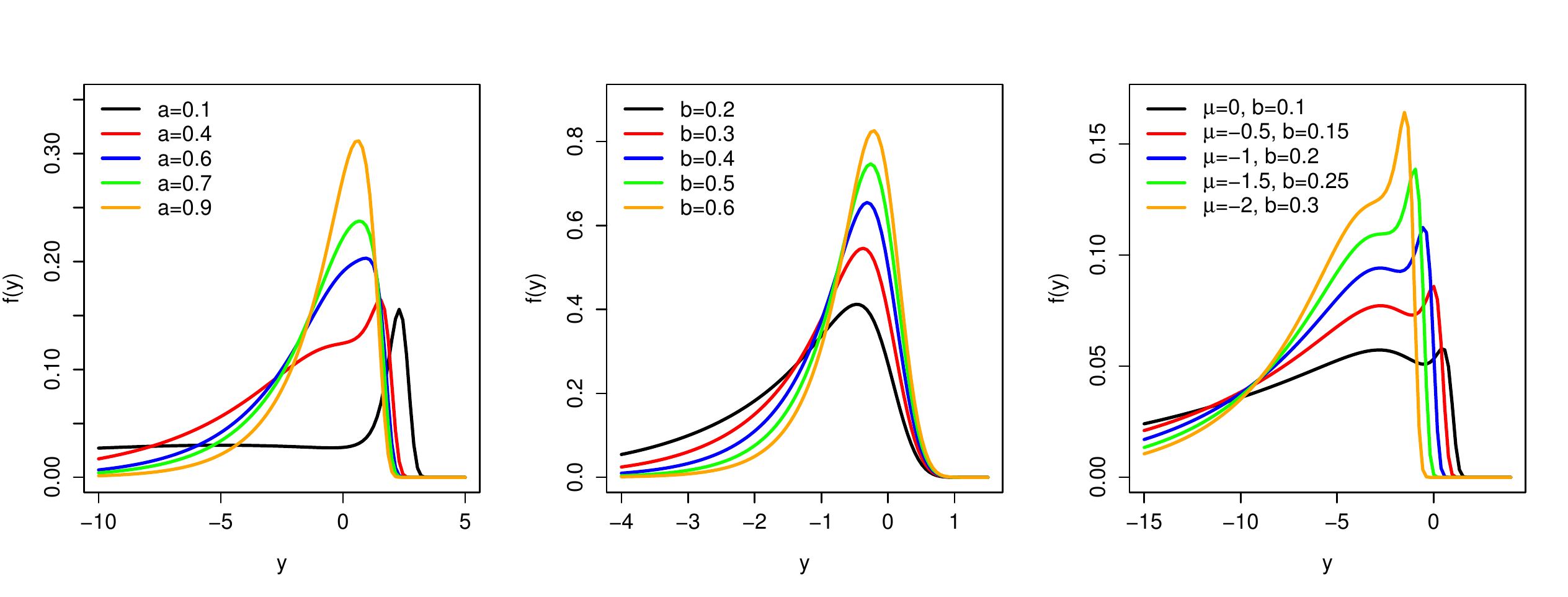}
\caption{Plots of the LEOLLW pdf. (a) For $\mu=1$, $\sigma=0.5$ and $b=0.3$. (b) For $\mu=0$, $\sigma=0.5$ and $a=0.5$. (c) For $\sigma=0.5$ and $a=0.4$.}
\label{fdp_y}
\end{figure}

The survival function of $Y$ is
\begin{eqnarray}\label{survival}
S(y) = 1- \frac{\left[1-\exp\left\{-\exp\left(\frac{y-\mu}{\sigma}\right)\right\}\right]^{ab}}
{\left\{\left[1-\exp\left\{-\exp\left(\frac{y-\mu}{\sigma}\right)\right\}\right]^a+\exp\left[-a\exp\left(\frac{y-\mu}{\sigma}\right)\right]\right\}^{b}}.
\end{eqnarray}

The pdf of $Z=(Y-\mu)/\sigma$ has the form
\begin{eqnarray}\label{densZ}
\pi(z)
=
\frac{
ab \exp\left\{z-a\exp\left(z\right)\right\}
\left[1-\exp\left\{-\exp\left(z\right)\right\}\right]^{ab-1}
}
{
\big\{\left[1-\exp\left\{-\exp\left(z\right)\right\}\right]^a+\exp\left[-a\exp\left(z\right)\right]\big\}^{b+1}
},
\quad z \in \mathbf{R}.
\end{eqnarray}

Some properties of $Y$ are reported below:
\begin{itemize}
\item[(PL1)]
Equation (\ref{fy}) gives $\lim_{y\to\infty} f(y)=0$. Moreover,
since  $f(y)=\pi\big((y-\mu)/\sigma\big)/\sigma$ and $\pi\big((y-\mu)/\sigma\big)=\exp[(y-\mu)/\sigma]\,f(\exp[(y-\mu)/\sigma])$,
where $f(x)$ is the pdf \eqref{pdf_eollw} with $\alpha=\lambda=1$.
By \eqref{lim-exp-weibull} and the L'Hôpital's rule, we obtain
\begin{align*}
\lim_{y\to-\infty} f(y)
=
{1\over\sigma}
\lim_{y\to-\infty}
\exp\biggl({y-\mu\over\sigma}\biggr) f\biggl(\exp\Big({y-\mu\over\sigma}\Big)\biggr)
=0.
\end{align*}
\item[(PL2)]
By Property (P5), the stochastic representation for $Y$ holds
\begin{align*}
Y=\mu+\sigma\log\bigl[\log(1+D)\bigr],
\quad
D\sim {\rm DAGUM}(a,1,b).
\end{align*}

\item[(PL3)]
Let $\mu=0$ and $\sigma=1$. A simple calculation gives
\begin{align*}
\pi'(z)
=
\pi(z)
\left\{
1+\widetilde{\omega}\left[
1-{\mathbbm{H}_0(\widetilde{\omega})\over 1-\exp\left(-\widetilde{\omega}\right)}
\right]
\right\},
\quad \widetilde{\omega}=\exp(z),
\end{align*}
where $\mathbbm{H}_0$ is as in \eqref{G-H-def}.
Then, the critical points of the pdf of $Y$ are the roots of:
\begin{align*}
\biggl({1\over\widetilde{\omega}}+1\biggr)[1-\exp\left(-\widetilde{\omega}\right)]
=
\mathbbm{H}_0(\widetilde{\omega}).
\end{align*}
\item[(PL4)]
Note that the function $\pi$ can be expressed as
\begin{eqnarray}\label{densZ-1}
\pi(z)
=
\frac{
b \left[1-F_G(-z)\right]^{ab-1} f_G(-z-\log(a))
}
{
\big\{\left[1-F_G(-z)\right]^a + F^a_G(-z)\big\}^{b+1}
},
\quad  z \in \mathbf{R},
\end{eqnarray}
where $F_G(x)=\exp[-\exp(-x)]$ is the standard Gumbel distribution and
$f_G(x)=\exp\{-[x+\exp(-x)]\}$, $x \in \mathbf{R}$, is the associated pdf.
The extreme-value standard distribution (Gumbel law) is a special case of (\ref{densZ}) when $a=b=1$.
\end{itemize}

\subsection{Another representation for the LEOLLW distribution}

First, we define the pdf $g(u)$ as follows
\begin{eqnarray}\label{densZ-2}
g(u)
=
\frac{
b \left[1-F_G(-S(u))\right]^{ab-1}
}
{
\big\{\left[1-F_G(-S(u))\right]^a + F^a_G(-S(u))\big\}^{b+1}
},
\quad 0 < u < 1,
\end{eqnarray}
where $S(u)=-F_G^{-1}(1-u)-\log(a)$ and $F_G^{-1}$ denotes the inverse of the
standard Gumbel cdf. For $a=b=1$, $g(u)$ reduces to the continuous uniform
distribution in the interval $(0,1)$. The cdf corresponding to $g(u)$ is
\begin{eqnarray*}
G(u)
=
\frac{
\left[1-F_G(-S(u))\right]^{ab}
}
{
\big\{\left[1-F_G(-S(u))\right]^a + F^a_G(-S(u))\big\}^{b}
}.
\end{eqnarray*}
	
Second, let $A:\mathbf{R}\to (0,1)$ be an one-to-one and monotone transformation.
So, we define the following cdf
\begin{align*}
F(z)=\int_{0}^{A(z)}g(u){\rm d}u=G(A(z)).
\end{align*}
Setting $A(z)=1-F_G(-z-\log(a))$, where $F_G$ is the standard Gumbel distribution,
we have
$F(z)=G\big(1-F_G(-z-\log(a))\big)$,
and
\begin{align*}
F'(z)=\pi(z), \quad  z \in \mathbf{R},
\end{align*}
where $\pi(z)$ is as in \eqref{densZ-1}.
	
The previous arguments show the following result since $F_G^{-1}(p)=-\log[-\log(p)]$ (for $0<p<1$).

\begin{proposition}
The random variable $Z \sim {\rm LEOLLW}(0,1,a,b)$ has the stochastic representation:
\begin{align*}
Z=-F_G^{-1}(1-U)-\log(a)=\log[-\log(1-U)]-\log(a),
\end{align*}
where $U$ is distributed according to \eqref{densZ-2}.
\end{proposition}
As a consequence of the proposition above, we obtain
\begin{itemize}
\item
The random variable  $Y \sim {\rm LEOLLW}(\mu,\sigma,a,b)$ admits the stochastic representation
\begin{align*}
Y=\mu-\sigma\big[F_G^{-1}(1-U)+\log(a)\big]
=
\mu-\sigma\log(a)+\sigma\log[-\log(1-U)].
\end{align*}
\item
The random variable  $X\sim \mbox{EOLLW}(\alpha,\lambda,a,b)$ can be expressed as
\begin{align*}
X=\lambda\exp\left\{-{1\over\alpha}\,\big[F_G^{-1}(1-U)+\log(a)\big]\right\}
=
\lambda a^{-1/\alpha}
\big[-\log(1-U)\big]^{1/\alpha}.
\end{align*}
\end{itemize}

\subsection{Modality of the LEOLLW density}
Since $\pi(z)=\sigma f(y;\mu,\sigma,a,b)$, $y=\sigma z+\mu$,
the shape of the LEOLLW pdf $f(y)$ is uniquely determined by the shape
of the pdf $\pi(z)$. Then, for simplicity, we consider the analysis
of modality of $\pi(z)$, i.e., when $\mu=0$ and $\sigma=1$.
	
By property (PL3), a critical point of the LEOLLW density satisfies
the following equation:
\begin{align*}
\mathbbm{I}_0(\widetilde{\omega})=
\mathbbm{H}_0(\widetilde{\omega}),
\quad
\widetilde{\omega}=\exp(z),
\end{align*}
where
\begin{align*}
\mathbbm{I}_0(\widetilde{\omega})
=
\biggl({1\over \widetilde{\omega}}+1\biggr)[1-\exp\left(-\widetilde{\omega}\right)]
\end{align*}
and
$\mathbbm{H}_0$ is as in \eqref{G-H-def}.
The function $\mathbbm{H}_0$ was addressed in Section \ref{Modality-EOLLW}.
The function $\mathbbm{I}_0$ is unimodal with mode $\widetilde{\omega}\approx 1.79328$. Moreover, $\lim_{\widetilde{\omega}\to 0}\mathbbm{I}_0(\widetilde{\omega})=\lim_{\widetilde{\omega}\to \infty}\mathbbm{I}_0(\widetilde{\omega})=1$.
	
The following result shows that, regardless of the choice of the parameters,
a critical point of the LEOLLW pdf always exists.

\begin{proposition}
Equation
$\mathbbm{I}_0(\widetilde{\omega})=
\mathbbm{H}_0(\widetilde{\omega})$ has at least one root on $(0,\infty)$.
\end{proposition}
\begin{proof}
Since $\lim_{\widetilde{\omega}\to 0^+}\{\mathbbm{I}_0(\widetilde{\omega})-\mathbbm{H}_0(\widetilde{\omega})\}=a$
and $\lim_{\widetilde{\omega}\to \infty}\{\mathbbm{I}_0(\widetilde{\omega})-\mathbbm{H}_0(\widetilde{\omega})\}=-a$,
the proof follows by using the intermediate value theorem.
\end{proof}
	
Since $\mathbbm{I}_0$ is unimodal and $\mathbbm{H}_0$ is strictly concave and increasing for $0<a\leqslant 1$, see \eqref{shape-H}, and  $\mathbbm{G}_0-\mathbbm{H}_0$ is not a periodic function, it is natural
to expect that Equation $\mathbbm{I}_0(\widetilde{\omega})
=\mathbbm{H}_0(\widetilde{\omega})$ has at most three positive roots (but we do not have a proof). In what remains in this section, we suppose that $\mathbbm{I}_0-\mathbbm{H}_0$ has at most three zeros. Then, we have the
following possible scenarios:
\begin{itemize}
\item
If $\mathbbm{I}_0$ and $\mathbbm{H}_0$ do not have a point of intersection then the LEOLLW pdf has no critical points. Property (PL1) implies that the pdf is the
zero function, which is absurd, so this scenario cannot occur.
\item
If $\mathbbm{I}_0$ and $\mathbbm{H}_0$ have a single point of intersection, then
the LEOLLW pdf has a unique critical point. By Property (PL1), it follows
that the {\rm EOLLW} pdf is unimodal.
\item
If $\mathbbm{I}_0$ and $\mathbbm{H}_0$ do have two points of intersection,
then the LEOLLW pdf has two critical points. But, this is an absurd by Property (PL1),
so this scenario cannot occur either.
\item
Further, if $\mathbbm{I}_0$ and $\mathbbm{H}_0$ have three points of intersection,
then the {\rm EOLLW} pdf has three critical points, say $z_0,z_1,z_2$. Suppose that $z_0<z_1<z_2$. Again, Property (PL1) gives that $z_0$ and $z_2$ are maximum points
and $z_1$ is a minimum point of the {\rm EOLLW} pdf, thus ensuring the bimodality.
\end{itemize}
	
Hence, we have established the following result:
\begin{theorem}\label{teo-modality}
If $Z \sim{\rm LEOLLW}(0,1,a,b)$ with $0<a\leqslant 1$,
then the pdf \eqref{densZ} of $Z$ is uni- or bimodal.
\end{theorem}
	
It is well-known that the standard Gumbel distribution (LEOLLW with $a=b=1$)
is unimodal, which is compatible with Theorem  \ref{teo-modality}.

\subsection{Tail behavior of LEOLLW}	
In this subsection, by following Definition \ref{def-tail}, we prove that the LEOLLW model has upper light-tail distribution, but the lower tail does not have a well-defined behavior when it is compared with the tail of an exponential distribution. That is, the lower tail of LEOLLW is neither least light-tail nor least heavy-tail.
\begin{proposition}
	The {\rm LEOLLW} has upper light-tail distribution and  the limit
	\begin{align}\label{limit-to-prove}
	\lim_{y\to -\infty} {\exp(ty)\over 1-S(y)} \
	\text{depends on} \ t>0.
	\end{align}
	Here, $S(y)$ is the survival function of $Y$ defined in
	\eqref{survival}.
\end{proposition}
\begin{proof}
	By Definition \ref{def-tail}, to prove that the {\rm LEOLLW} model has upper light-tail distribution, it must be verified that
	(for any $t > 0$)
	\begin{align}\label{limit-obj}
	\lim_{y\to \infty} {\exp(-ty)\over S(y)}=\infty.
	\end{align}
	
	Indeed, by Property (P5), $F(y)=\mathbbm{P}[D\leqslant T(\exp(y))]\eqqcolon F_D[T(\exp(y))]$, where $D\sim {\rm DAGUM}(a,1,b)$ (Dagum distribution Type I), $T(x)=\exp\{(x/\lambda)^\alpha\}-1$ with $\mu=\log(\lambda)$ and $\sigma=1/\alpha$, and
	$F_D[d]=(1+d^{-a})^{-b}$. Then,
	for any $t>0$, we have
	\begin{align}\label{ineq-in-1}
	\lim_{y\to \infty}
	{
		\exp(-ty)
		\over
		S(y)
	}
	=
	\lim_{y\to \infty}
	{
		\exp(-ty)
		\over
		\overline{F}_D[T(\exp(y))]
	}
	&=
	\lim_{y\to \infty}
	{
		\exp(-ty)
		\over
		1-[1+T(\exp(y))^{-a}]^{-b}
	}
	\eqqcolon
	L_1.
	\end{align}
	L'Hôpital's rule gives
	\begin{align}\label{hospital-1}
	L_1
	&=
	\lim_{y\to \infty}
	{
		t\exp(-ty)
		\over
		ab[1+T(\exp(y))^{-a}]^{-b-1} T(\exp(y))^{-a-1}T'(\exp(y))\exp(y)
	}
	\nonumber
	\\[0,2cm]
	&=
	{\lambda t\over ab\alpha}
	\lim_{y\to \infty}
	{
		[1+T(\exp(y))^{-a}]^{b+1} \,
		{
			\big\{\exp\big[({\exp(y)\over\lambda})^\alpha\big]-1\big\}^{a+1}
			\over
			\big({\exp(y)\over\lambda}\big)^{\alpha-1}
			\exp\big[({\exp(y)\over\lambda})^\alpha+(t+1)y\big]
		}
	}.
	\end{align}
	But, by considering the change of variables $z=\exp(y)/\lambda$,
	\begin{align}\label{hospital-2}
	\lim_{y\to \infty}
	{\big\{\exp\big[({\exp(y)\over\lambda})^\alpha\big]-1\big\}^{a+1}
		\over \big({\exp(y)\over\lambda}\big)^{\alpha-1}
		\exp\big[({\exp(y)\over\lambda})^\alpha+(t+1)y\big]}
	=
	{1\over\lambda^{t+1}}
	\lim_{z\to \infty}
	{[\exp(z^\alpha)-1]^{a+1}\over z^{\alpha+t} \exp(z^\alpha)}=\infty.
	\end{align}
	Since $\lim_{y\to \infty}T(\exp(y))=\infty$,  by combining \eqref{hospital-2} and \eqref{hospital-1} with \eqref{ineq-in-1}, the limit in \eqref{limit-obj} follows. This proves the upper light-tailedness of LEOLLW distribution.
	
	On the other hand, similarly to the steps done previously, by L'Hôpital's rule,
	\begin{align*}
	\lim_{y\to -\infty} {\exp(ty)\over 1-S(y)}
	&=
	\lim_{y\to -\infty}
	{
		\exp(ty)
		\over
		[1+T(\exp(y))^{-a}]^{-b}
	}
	\\[0,2cm]
	&=
	{\lambda t\over ab\alpha}
	\lim_{x\to \infty}
	{
		\big\{	\big\{\exp\big[({\exp(-x)\over\lambda})^\alpha\big]-1\big\}^{a}+1 \big\}^{b+1}
		\big\{\exp\big[({\exp(-x)\over\lambda})^\alpha\big]-1\big\}^{1-ab}
		\over
		\big({\exp(-x)\over\lambda}\big)^{\alpha-1}
		\exp\big[({\exp(-x)\over\lambda})^\alpha+(t-1)x\big]	
	}
	\\[0,2cm]
	&\eqqcolon
	{\lambda t\over ab\alpha}\, L_2,
	\end{align*}
	where the changing variables $x=-y$ has been considered. Moreover, by taking the change of variables $w=\exp(-x)/\lambda$, we have
	\begin{align*}
	L_2
	&	=
	{1\over \lambda^{1-t}}
	\lim_{w\to 0}
	{
		\big\{	\big[\exp(w^\alpha)-1\big]^{a}+1 \big\}^{b+1}
		\over
		\exp(w^\alpha)
	}
	\,
	{
		\big[\exp(w^\alpha)-1\big]^{1-ab}
		\over
		w^{\alpha-t} 	
	}
	\\[0,2cm]
	&=
	\begin{cases}
	\infty, & \alpha-t>0, \ 1-ab\leqslant 0;
	\\
	0, & \alpha-t<0, \ 1-ab\geqslant 0;
	\\
	{1\over \lambda^{1-t}}, & \alpha-t=0, \ 1-ab= 0;
	\\
	\infty, & \alpha-t=0, \ 1-ab< 0;
	\\
	0, & \alpha-t=0, \ 1-ab> 0;
	\\
	\infty, & \alpha-t<0, \ 1-ab< 0;
	\\
	0, & \alpha-t>0, \ 1-ab> 0.
	\end{cases}
	\end{align*}
	Therefore, we conclude that the limit in \eqref{limit-to-prove} depends on the choice of $t>0$.
\end{proof}

\section{The LEOLLW regression with censored data}\label{sec:reg}

The LEOLLW regression model is defined by
\begin{eqnarray}
\label{log-linear}
y_i=\mu_i+\sigma_i z_i, \quad i=1,\ldots,n.
\end{eqnarray}
Here, the random error $z_i$ has density \eqref{densZ}, and the location $\mu_i$ and dispersion $\sigma_i$ are related to the explanatory v\-a\-ria\-ble vector $\vn_i^\top=(v_{i1},\ldots,v_{ip})$ (for $i=1,\ldots,n$) by
\begin{eqnarray}\label{reg_1}
\mu_{i}={\vn_i^\top}\,\betn_1 \qquad \text{and} \qquad \sigma_{i}= \exp ({\vn_i^\top}\,\betn_2),
\end{eqnarray}
where $\betn_1=(\beta_{10},\ldots,\beta_{1p})^\top$  and  $\betn_2=(\beta_{20},\ldots,\beta_{2p})^\top$ are $p \times 1$ vectors of
unkown parameters.

Setting $\mun=(\mu_1,\ldots,\mu_n)^\top$,
$\boldsymbol{\sigma}=(\sigma_1,\ldots,\sigma_n)^\top$,
and $\Vn=(\vn_1, \ldots, \vn_n)^\top$ for a known matrix, we can write $\mun=\Vn\betn_{1}$ and $\boldsymbol{\eta}=\exp(\boldsymbol{\sigma})=\Vn\betn_{2}$,
where $a$ and $b$ denote shape parameters for the regression.

Equation \eqref{log-linear} gives the log-odd log-logistic Weibull (LOLLW)
regression for $a=1$, log-exponentiated Weibull (LEW) regression for $b=1$,
and  log-Weibull (LW) regression for $a=b=1$.

The survival function of $Y_i|\vn$ follows from Equations \eqref{log-linear} and \eqref{reg_1}
\begin{eqnarray}\label{Sy_x}
S(y_i|\vn)= 1- \frac{\left[1-\exp\left\{-\exp\left(\frac{y_i-\mu_{i}}{\sigma_i}\right)\right\}\right]^{a\,b}}
{\left\{\left[1-\exp\left\{-\exp\left(\frac{y_i-\mu_{i}}{\sigma_i}\right)\right\}\right]^a+\exp\left[-a\exp\left(\frac{y_i-\mu_{i}}{\sigma_i}\right)\right]\right\}^{b}}.
\end{eqnarray}

\subsection{Estimation}

Consider $n$ observations $(y_1,\vn_1),\ldots, (y_n,\vn_n)$,
where $y_i=\min\{\log(X_i),\log(C_i)\}$. The logarithm of the likelihood function
for $\tetn=(a,b,\betn_{1}^{\top}, \betn_{2}^{\top})^\top$ (assuming right censoring)
has the form
\begin{eqnarray}\label{verossim1}
l(\tetn)&=&r\log\left(\frac{a\,b}{\sigma_i}\right)\sum_{i \in F}z_i-a\sum_{i \in F}\exp(z_i)+(a\,b-1)\sum_{i \in F}\log\left\{1-\exp[-\exp(z_i)]\right\}-\nonumber\\&&
(b+1)\sum_{i \in F}\log\left\{[1-\exp\{-\exp(z_i)\}]^a+\exp[-a\exp(z_i)]\right\}+\nonumber\\&&
\sum_{i \in C}\log
\left\{1- \frac{\left[1-\exp\left\{-\exp\left(z_i\right)\right\}\right]^{a\,b}}
{\left\{\left[1-\exp\left\{-\exp\left(z_i\right)\right\}\right]^a+\exp\left[-a\exp\left(z_i\right)\right]\right\}^{b}}\right\},
\end{eqnarray}
 where $r$ is the number of uncensored observations (failures) and $z_i=(y_i-\mu_i)/\sigma_i$. Here, $F$ and $C$ are the sets for the
 uncensored individuals and individuals with right censoring,
 respectively.

Equation (\ref{verossim1}) can be maximized using SAS (Proc NLMixed) or {\tt R} (\textsc{optim}, \textsc{gamlss}) (R Development Core Team, 2022),
among others, with initial values for $\betn_1$ and $\betn_2$ equal to
those from the fit of the LW regression ($a=b=1$).

\subsection{Simulations study}

We perform some simulations in order to evaluate the accuracy of the MLEs. We obtain 1,000 random samples from the LEOLLW $(\mu, \sigma, a, b)$ model using \textsc{optim} package in {\tt R} for sample sizes $n=100,250$ and $500$ and censoring percentages approximately equal to $0\%$, $15\%$ and $45\%$. For each configuration, the log-lifetimes $\text{log}(x_1), \ldots , \text{log}(x_n)$ are generated from (\ref{log-linear}) with two covariates, i.e., ${\vn_i^\top}\betn=\beta_{01}+\beta_{11}v_1+\beta_{12}v_2$, where $v_1 \sim \text{uniform} (0, 1)$ and $v_{2} \sim \text{Binomial} (1,0.5)$. The censoring times $c_1,\ldots , c_n$ are generated from a uniform distribution $(0, \tau)$, where $\tau$ controls the censoring percentage. The true parameter values are: $\beta_{10}=3,\beta_{11}=2.5,\beta_{12}=1.9,\sigma=0.3,a=0.5$ and $b=0.9$.
The simulation process is given below:

(i) Generate $v_{i1} \sim \text{uniform} (0,1)$ and $v_{i2} \sim \text{Binomial} (1,0.5)$;

(ii) Generate $z_i \sim \text{LEOLLW}(0,1,a,b)$ (\ref{densZ});

(iii) Calculate $y_{i}^{*} = \beta_{0}+\beta_{1} v_{i1} +\beta_2 v_{i2} + \sigma\, z_i$;

(iv) Generate $c_{i} \sim \text{uniform} (0,\tau)$;

(v) Calculate the survival times $y_i=\text{min}(y_{i}^{*},c_i)$;

(vi) If $y_{i}^{*} <c_i$, then $\delta_i = 1$; otherwise, $\delta_i = 0$, for $i = 1,\ldots, n$.

Table \ref{simulation1} reports the average estimates (AEs), biases
and mean square errors (MSEs) of the MLEs of the parameters. These results show
that the AEs converge to the true parameters and the biases
and MSEs decrease when $n$ increases, thus indicating consistent
estimators. The empirical coverage probabilities (CPs) of the
parameters for the 95\% confidence interval given in Table \ref{cps}
also reveal that the CPs tend to the confidence level.

\begin{table}[ht]
\centering
\caption{Simulatins results from the fitted LEOLLW regression.}
{\scriptsize{
\begin{tabular}{rr@{\hskip 0.4in}lrr@{\hskip 0.4in}rrr@{\hskip 0.4in}rrrr}
  \hline
   &  & & $n=100$ &  & & $n=250$ &  &  & $n=500$ &  \\
  \cline{3-5} \cline{6-8} \cline{9-11}
  $\%$ & $\theta$ & AEs & Biases & MSEs & AEs & Biases & MSEs & AEs & Biases & MSEs \\
   \hline
\multirow{6}{*}{$0\%$} & $\beta_{10}$ & 2.9621 & -0.0379 & 0.0599 & 2.9776 & -0.0224 & 0.0292 & 2.9882 & -0.0118 & 0.0168 \\
   &  $\beta_{11}$ & 2.4958 & -0.0042 & 0.0226 & 2.4969 & -0.0031 & 0.0082 & 2.4990 & -0.0010 & 0.0033 \\
  &  $\beta_{12}$ & 1.9025 & 0.0025 & 0.0070 & 1.8997 & -0.0003 & 0.0027 & 1.9025 & 0.0025 & 0.0013 \\
  &  $\sigma$ & 0.2898 & -0.0102 & 0.0087 & 0.2940 & -0.0060 & 0.0042 & 0.2967 & -0.0033 & 0.0022 \\
  & $a$ & 0.4598 & -0.0402 & 0.0202 & 0.4719 & -0.0281 & 0.0065 & 0.4872 & -0.0128 & 0.0037 \\
& $b$ & 1.0057 & 0.1057 & 0.1968 & 0.9646 & 0.0646 & 0.1006 & 0.9325 & 0.0325 & 0.0567 \\
   \hline

   \multirow{6}{*}{$10\%$} & $\beta_{10}$  & 2.9387 & -0.0613 & 0.0617 & 2.9660 & -0.0340 & 0.0309 & 2.9868 & -0.0132 & 0.0184 \\
  & $\beta_{11}$ & 2.4934 & -0.0066 & 0.0258 & 2.4977 & -0.0023 & 0.0099 & 2.4995 & -0.0005 & 0.0039 \\
    & $\beta_{12}$ & 1.9017 & 0.0017 & 0.0074 & 1.9014 & 0.0014 & 0.0030 & 1.9026 & 0.0026 & 0.0015 \\
    & $\sigma$ & 0.3021 & 0.0021 & 0.0088 & 0.2974 & -0.0026 & 0.0041 & 0.2964 & -0.0036 & 0.0025 \\
   & $a$ & 0.4687 & -0.0313 & 0.0228 & 0.4697 & -0.0303 & 0.0069 & 0.4856 & -0.0144 & 0.0039 \\
  & $b$ & 1.0497 & 0.1497 & 0.2115 & 0.9870 & 0.0870 & 0.1060 & 0.9369 & 0.0369 & 0.0621 \\

   \hline
\multirow{6}{*}{$30\%$} & $\beta_{10}$ & 2.9078 & -0.0922 & 0.0785 & 2.9498 & -0.0502 & 0.0338 & 2.9802 & -0.0198 & 0.0228 \\
 & $\beta_{11}$ & 2.4884 & -0.0116 & 0.0379 & 2.4960 & -0.0040 & 0.0133 & 2.4986 & -0.0014 & 0.0057 \\
  & $\beta_{12}$ & 1.8956 & -0.0044 & 0.0113 & 1.8996 & -0.0004 & 0.0040 & 1.9020 & 0.0020 & 0.0020 \\
  & $\sigma$ & 0.3166 & 0.0166 & 0.0119 & 0.3021 & 0.0021 & 0.0043 & 0.2967 & -0.0033 & 0.0029 \\
  & $a$ & 0.4718 & -0.0282 & 0.0351 & 0.4664 & -0.0336 & 0.0097 & 0.4813 & -0.0187 & 0.0049 \\
  & $b$ & 1.1274 & 0.2274 & 0.2771 & 1.0185 & 0.1185 & 0.1198 & 0.9519 & 0.0519 & 0.0764 \\

   \hline
\end{tabular}
\label{simulation1}}}
\end{table}

\begin{table}[ht]
\centering
\caption{CPs for the 95\% nominal level from the fitted LEOLLW regression.}
{\scriptsize{
\begin{tabular}{c@{\hskip 0.4in}ccc@{\hskip 0.4in}ccc@{\hskip 0.4in}ccc}
  \hline
  & \multicolumn{3}{c}{$0\%$}  & \multicolumn{3}{c}{$10\%$}  & \multicolumn{3}{c}{$30\%$} \\
    \cline{2-10}
$n$ & $100$ &$250$ &$500$   & $100$ &$250$ &$500$ & $100$ &$250$ &$500$  \\
 \cline{2-10}
$\beta_{10}$ & 0.953 & 0.951 & 0.953 & 0.976 & 0.962 & 0.957  & 0.980 & 0.973 & 0.960  \\
  $\beta_{11}$ & 0.917 & 0.934 & 0.967  & 0.930 & 0.924 & 0.960 & 0.927 & 0.929 & 0.945 \\
  $\beta_{12}$ & 0.922 & 0.928 & 0.935 & 0.941 & 0.932 & 0.945 & 0.932 & 0.933 & 0.940 \\
  $\sigma$ & 0.909 & 0.926 & 0.949 & 0.954 & 0.950 & 0.947 & 0.976 & 0.979 & 0.958\\
  $a$ & 0.839 & 0.891 & 0.914  & 0.845 & 0.879 & 0.911 & 0.822 & 0.860 & 0.899 \\
  $b$ & 0.961 & 0.948 & 0.951   & 0.974 & 0.967 & 0.957 & 0.977 & 0.983 & 0.963 \\
   \hline
\end{tabular}
\label{cps}}}
\end{table}

\section{Residual Analysis}\label{sec:residuos}

\hspace{0.5cm} Residuals are important when determining the adequacy of a
regression model and detection of outliers. They play a crucial role in validating the regression by examining the residual plots; see, for example, Cox and Snell (1968), Cook and Weisberg (1982), Collet (2003), Ortega et al. (2008), Silva et al. (2011) and recently Hashimo et al. (2021).\\

{\bf Martingale residuals}\\

We adopt the martingale residuals
$r_{M_i}=\delta_i+\log[S(y_i;\widehat{\tetn})]$ (Fleming and Harrington, 1991),
where $\delta_i$ is the censoring indicator, i.e., $\delta_i=0\,\, (\delta_i=1)$
is a censored (uncensored) observation and $S(y_i;\widehat{\tetn})$ is
the estimated survival function (\ref{Sy_x}). By setting  $\hat{z}_i=(y_{i}-\xn^T_i\widehat{\betn})/\widehat{\sigma}$,
these residuals (under right censoring) become
\begin{eqnarray}\label{martingale_right}
r_{M_i}
=
\left\{
\begin{array}{ll}
1+\log\left\{
1- \frac{\left[1-\exp\left\{-\exp\left(\hat{z}_i\right)\right\}\right]^{\hat{a}\,\hat{b}}}
{\left\{\left[1-\exp\left\{-\exp\left(\hat{z}_i\right)\right\}\right]^{\hat{a}}+\exp\left[-\hat{a}\exp\left(\hat{z}_i\right)\right]\right\}^{\hat{b}}}\right
\}&
\mbox{if}\,\,\,\, \delta_i=1,
\\\\
\log\left\{
1- \frac{\left[1-\exp\left\{-\exp\left(\hat{z}_i\right)\right\}\right]^{\hat{a}\,\hat{b}}}
{\left\{\left[1-\exp\left\{-\exp\left(\hat{z}_i\right)\right\}\right]^{\hat{a}}+\exp\left[-\hat{a}\exp\left(\hat{z}_i\right)\right]\right\}^{\hat{b}}}
\right\}&
\mbox{if}\,\,\,\, \delta_i=0.
\end{array} \right.
\end{eqnarray}\\

{\bf Modified deviance residuals}\\

\hspace{0.5cm}The deviance component residuals (Therneau et al., 1990; Collett, 2003)
can be expressed as
\begin{displaymath}
r_{D_i}=\mbox{sgn}({r}_{M_i})\left\{-2\left[{r}_{M_i}+\delta_i\log\left(\delta_i-{r}_{M_i}\right)\right]\right\}^{1/2}.
\end{displaymath}
where $r_{M_i}$ is given by \eqref{martingale_right}.

\subsection{Simulations}

One thousand samples are generated based on each scenario of $n,\beta_{10},\beta_{11}, \beta_{12},\sigma,a,b$ and censoring rates from the previous simulation study.
After fitting the regression (\ref{log-linear}), we obtain the residuals $r_{M_i}$'s and $ r_{D_i}$'s. Figures \ref{rm_1}-\ref{rm_3} and \ref{rd_1}-\ref{rd_3}
provide the normal probability plots. They show that the empirical
distribution of the $r_{M_i}$'s is asymmetric around zero and presents
accentuated kurtosis. The residuals $r_{D_i}$'s have an empirical distribution
in good agreement with the standard normal distribution for lower censoring rates
and larger sample sizes.

\begin{figure}[htb!]
		\begin{center}
n=100 \hspace{4.5cm} n=250\hspace{4.5cm} n=500 \\
			\includegraphics[width=16cm,height=4.5cm]{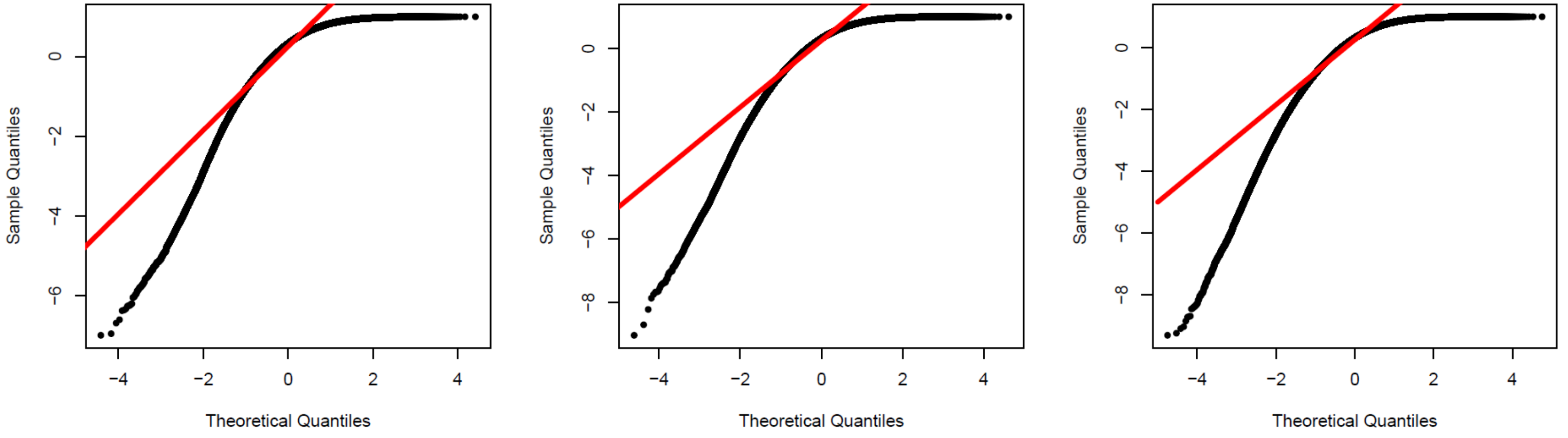}
			\caption{Normal probability plots for $r_{M_i}$'s without censoring}
			\label{rm_1}
		\end{center}
	\end{figure}
\begin{figure}[htb!]
		\begin{center}
n=100 \hspace{4.5cm} n=250\hspace{4.5cm} n=500 \\
			\includegraphics[width=16cm,height=4.5cm]{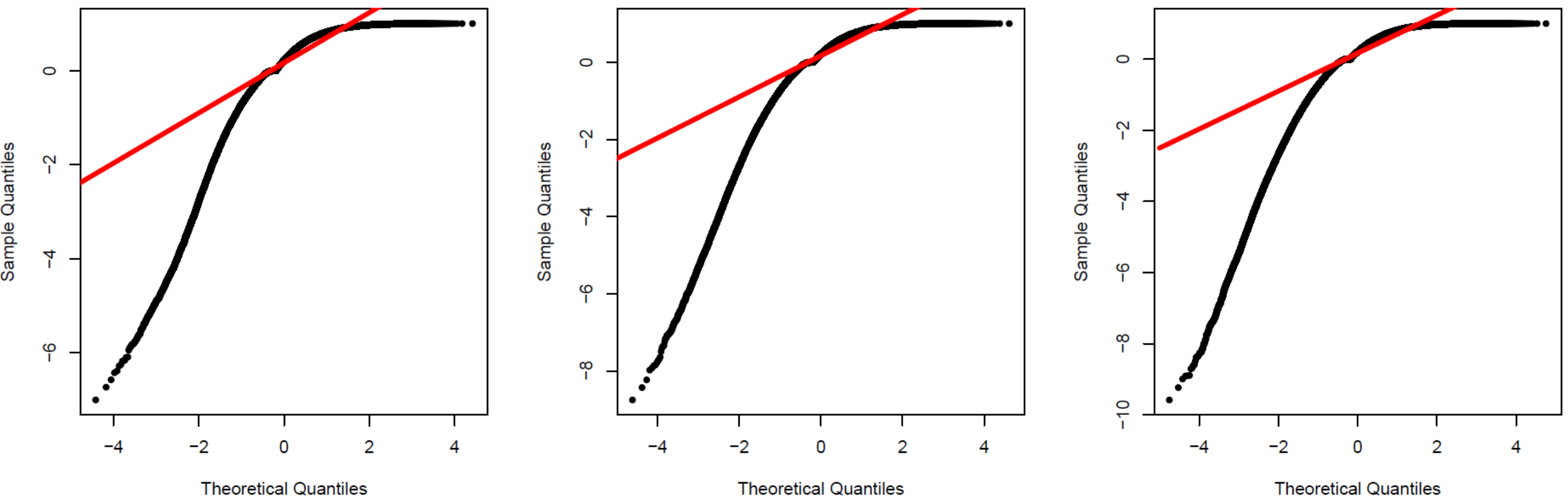}
			\caption{Normal probability plots for $r_{M_i}$'s with censoring rate $10\%$}
			\label{rm_2}
		\end{center}
	\end{figure}

\begin{figure}[htb!]
		\begin{center}
n=100 \hspace{4.5cm} n=250\hspace{4.5cm} n=500 \\
			\includegraphics[width=16cm,height=4.5cm]{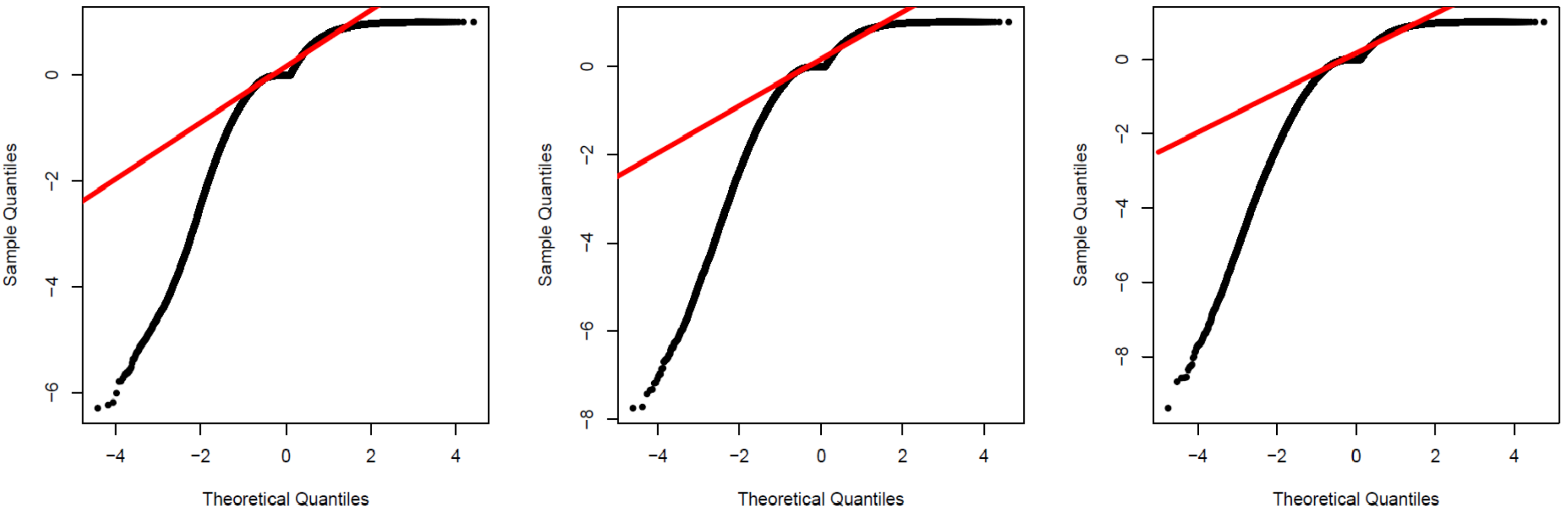}
			\caption{Normal probability plots for $r_{M_i}$'s with censoring rate $10\%$}
			\label{rm_3}
		\end{center}
	\end{figure}

\begin{figure}[htb!]
		\begin{center}
n=100 \hspace{4.5cm} n=250\hspace{4.5cm} n=500 \\
			\includegraphics[width=16cm,height=4.5cm]{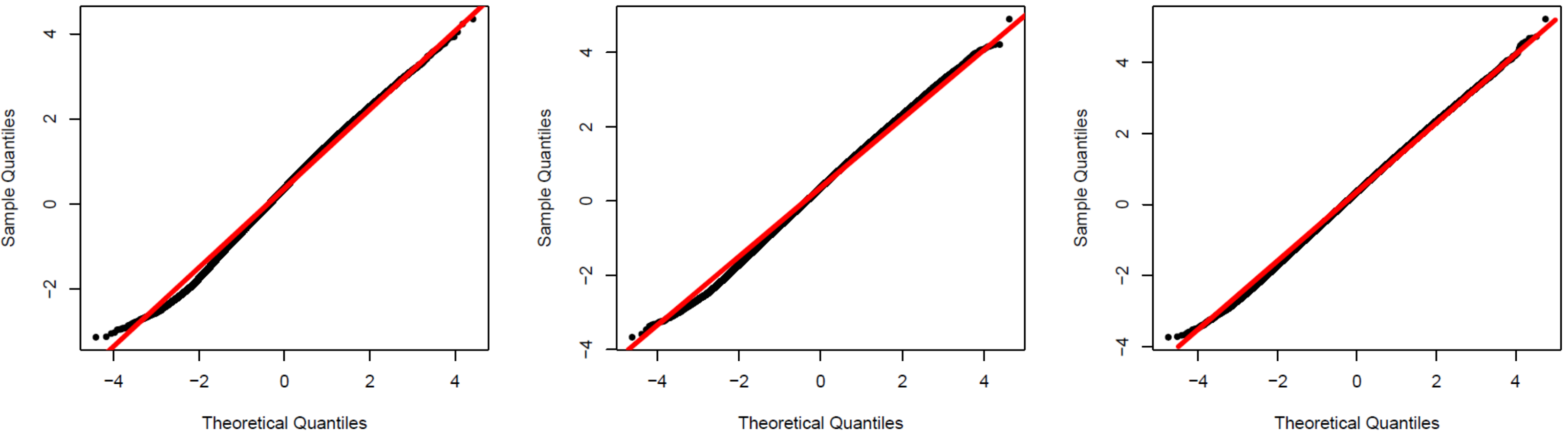}
			\caption{Normal probability plots for $r_{D_i}$'s without censoring}
			\label{rd_1}
		\end{center}
	\end{figure}
\begin{figure}[htb!]
		\begin{center}
n=100 \hspace{4.5cm} n=250\hspace{4.5cm} n=500 \\
			\includegraphics[width=16cm,height=4.5cm]{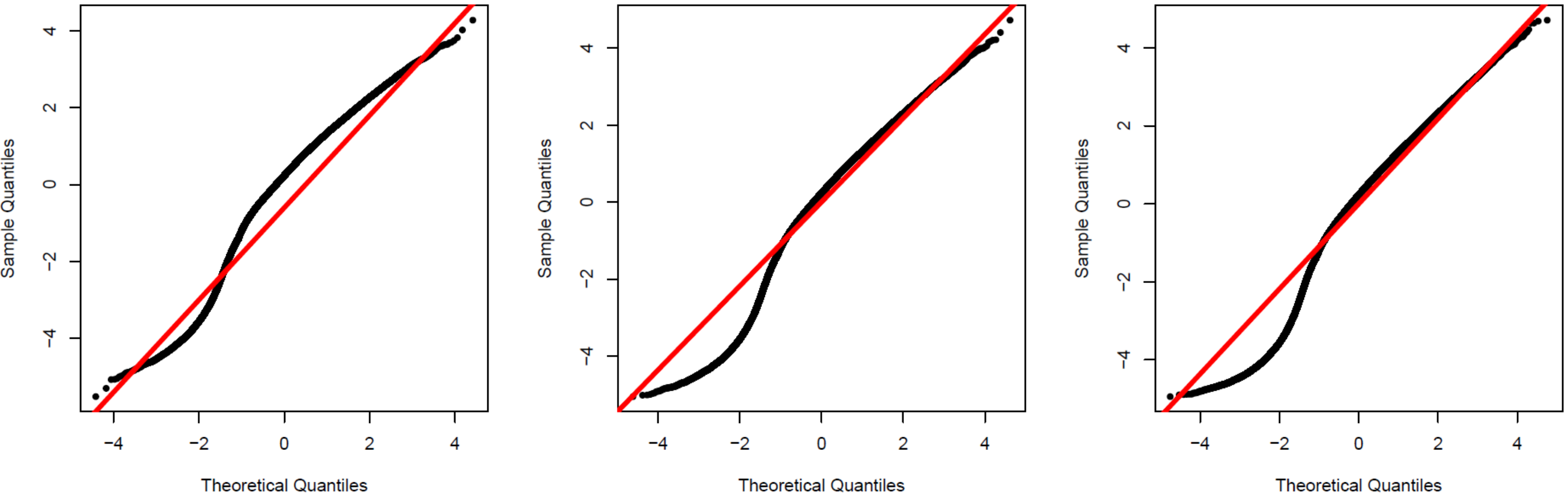}
			\caption{Normal probability plots for $r_{D_i}$'s with censoring rate $10\%$}
			\label{rd_2}
		\end{center}
	\end{figure}

\begin{figure}[htb!]
		\begin{center}
n=100 \hspace{4.5cm} n=250\hspace{4.5cm} n=500 \\
			\includegraphics[width=16cm,height=4.5cm]{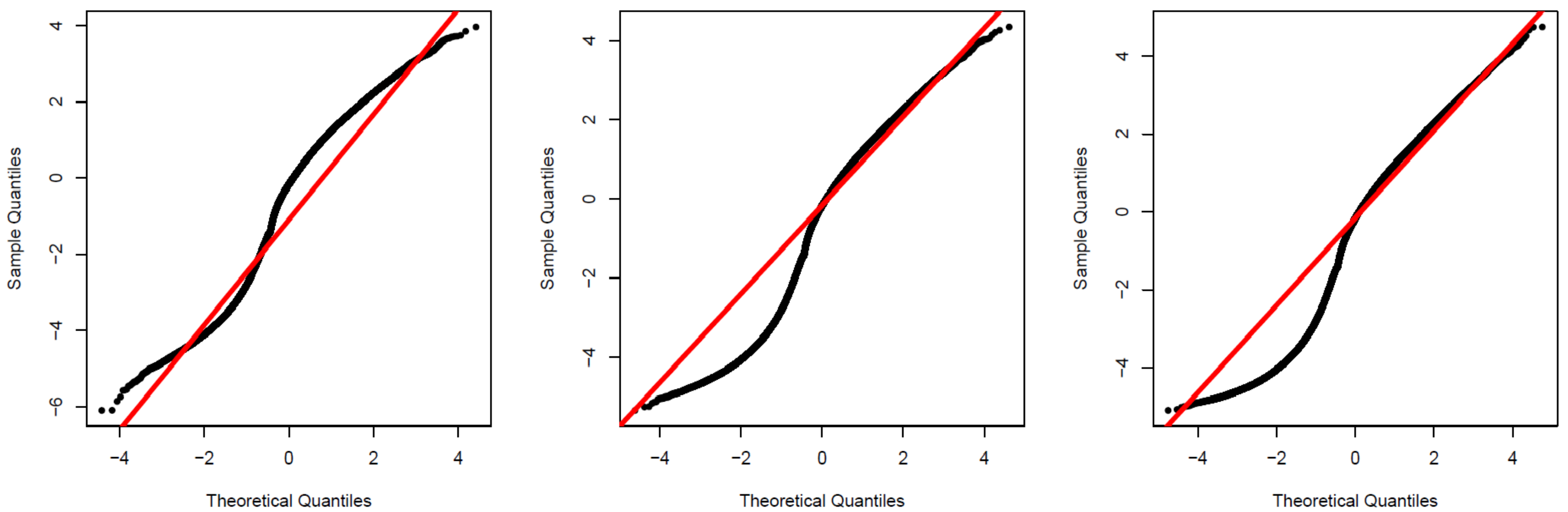}
			\caption{Normal probability plots for $r_{D_i}$'s with censoring rate $30\%$}
			\label{rd_3}
		\end{center}
	\end{figure}

\section{Application: Japanese-Brazilian emigration data}\label{sec:application}

We compare the fits of the LW, LEW, LOLLW and LEOLLW models by calculating
the MLEs, their standard errors (SEs) and the values of the Akaike Information Criterion (AIC), Consistent Akaike, Information Criterion (CAIC), and Bayesian Information Criterion (BIC) using the \textsc{gamlss} package in {\tt R} software ({\tt R} Core Team, 2022).

Based on technological development and economic growth in the mid-1980s,
Japan began to attract many immigrants from Brazil with Japanese ancestry. This phenomenon intensified after June 1990, and at the end of the 1990s these
Brazilians formed the third largest community of foreigners living in Japan, with approximately 312,979 people in 2010, behind only Koreans and Chinese (Kawamura, 1999). However, afterward various global crises severely affected Japan, with negative repercussions on the insertion of immigrants in the labor market, leading to the need for professional retraining as a condition for remaining in Japan, unlike the initial situation where high qualification was not required. In response to a request from the Japanese government to the Brazilian government in 2008, Federal University of Mato Grosso (UFMT), by means of the Brazilian Open University program (UAB), together with Tokai University, locted in the city of Hiratsu, Japan, began offering a teacher training course in the distance learning modality. The course, which lasts 4 years, began in 2009, with the aim of qualifying 300 Brazilian teachers working in Japan to work in Brazilian and Japanese schools. Thus, this study seeks to identify the factors that influence the time spent in Japan of the students of the teacher training course offered by UFMT/Tokai, because it is known that the length of stay can be affected by covariables, which are extremely important to the model used in this analysis. The data were obtained by an electronic survey (Babbie, 1999) with the objective to get
the characteristics, actions and/or opinions of the group of students using the internet as a learning tool. The survey was conducted in the first school semester of 2010 by means of a reserved site with access only by students, for which 246
completed questionnaires were received. Of these, only 150 were used for analysis because of responses by students of other nationalities. We consider the time (in years) spent in Japan as a response variable counted from the arrival
date until July 2012, with censoring of students who returned to Brazil at least
once. The variables under study are:

\begin{itemize}
\itemsep-0.3em

\item $y_i$: time spent in Japan (in years);
\item $\delta_{i}$: the censoring indicator (0 = censored, 1 = failure);
\item $v_{i1}:$ sex (0 = female, 1 = male);
\item $v_{i2}:$ age (in years);
\item $v_{i3}:$ reason for migration, (0=accompany family, 1=better living conditions, 2=study, 3=new experiences/undeclared), defined by three dummy variables.
\end{itemize}

We fit the regression
$$
y_i=\mu_{i}+\sigma_{i} z_i, \,\,\,\,\, i=1, \ldots, n,$$
where the response variable follows the LEOLLW distribution in (\ref{fy}), and the systematic components are (for $i=1,\ldots,150$)
$$\mu_{i}=\beta_{10}+\beta_{11}v_{i1}+\beta_{12}v_{i2}+\beta_{13}v_{i3}+\beta_{14}v_{i4}+\beta_{15}v_{i5}$$
and
$$\sigma_{i}=\exp(\beta_{20}+\beta_{21}v_{i1}+\beta_{22}v_{i2}+\beta_{23}v_{i3}+\beta_{24}v_{i4}+\beta_{25}v_{i5}).$$

The initial values for $\betn_{1}^{\top}$ and $\betn_{2}^{\top}$ were obtained
from the fitted LW regression.  Table \ref{aic} reports the values of the previous statistics for some fitted models, which indicate that the LEOLLW regression model can be chosen as the best model.

\begin{table}[ht]
\caption{Measures for some fitted regressions to the Japanese-Brazilian emigration data.} \label{aic} \centering  \vspace*{0.3cm}
\begin{tabular}{rrrr}
  \hline\hline
Regression & AIC & BIC & CAIC \\
  \hline \hline
LEOLLW & 172.22 & 214.37 & 168.26 \\
 LOLLW & 293.00 & 332.14 & 297.06 \\
  LExpW & 186.95 & 226.08 & 191.00 \\
  LW & 216.32 & 252.45 & 228.39 \\
   \hline \hline
\end{tabular}
\end{table}
A comparison of the proposed regression with some of its sub-models via
likelihood ration (LR) statistics is given in Table \ref{lr2}. So,
the LEOLLW regression gives a better fit to these data than
the other three sub-models.

\begin{table}[!htb]
\caption{LR statistics for the Japanese-Brazilian emigration data.}
\label{lr2}\centering {\vspace*{0.3cm}
\begin{tabular}{c|c|c|c}
\hline\hline
Model & Hypotheses & Statistic w & $p$-value \\ \hline\hline
LEOLLW vs LOLLW & $H_{0}:b=1$ vs $H_{1}: H_{0}\, \mbox{is false}$ & 122.787 & $<$0.00001 \\
LEOLLW vs LExpW & $H_{0}:a=1$ vs $H_{1}: H_{0}\, \mbox{is false}$ & 16.7293 & $<$0.00001 \\
LEOLLW vs Log-Weibull & $H_{0}:b=a=1$ vs $H_{1}: H_{0}\, \mbox{is false}$ & 48.1057 & $<$0.00001 \\ \hline\hline
\end{tabular}%
}
\end{table}

Table \ref{mle} gives the MLEs (and their SEs in parentheses) of
the parameters, which reveal that the covariates sex, age, and reasons for
migration (study and better living conditions) are significant
for the mean parameter $\mu$ at the 5\%. The three reasons for migration (study, better living conditions and new experiences) contribute to the dispersion of data.

\begin{table}[ht]
\centering
\caption{Results from the LEOLLW regression fitted to Japanese-Brazilian emigration data.} \vspace*{0.3cm}
\label{mle}
\begin{tabular}{rrrrr}
 \hline \hline
 & MLEs & SEs &   p-values \\
  \hline \hline
$\beta_{10}$ & 3.7403 & 0.1154 &  $<0.0001$ \\
  $\beta_{11}$ & -0.1093 & 0.0476 &  0.0230 \\
  $\beta_{12}$ & -0.0086 & 0.0026 & 0.0012 \\
  $\beta_{13}$ & -0.1262 & 0.0493 & 0.0115 \\
  $\beta_{14}$ & -0.1838 & 0.0615 &  0.0033 \\
 $\beta_{15}$ & -0.0902 & 0.0635 &  0.1580 \\
  $\beta_{20}$ & 2.6087 & 0.2917 &  $<0.0001$ \\
  $\beta_{21}$ & -0.0461 & 0.1157 &  0.6906 \\
  $\beta_{22}$ & -0.0675 & 0.0065 &  $<0.0001$ \\
  $\beta_{23}$ & -0.3974 & 0.1005 &  0.0001 \\
 $\beta_{24}$ & -0.4286 & 0.1548 &  0.0064 \\
  $\beta_{25}$ & -0.3343 & 0.1596 &  0.0379 \\
  $\text{log}(a)$ & 3.0955 & 0.0758 &  \\
  $\text{log}(b)$ & -2.9711 & 0.0789 &  \\
\hline \hline
\end{tabular}
\end{table}

From the fitted LEOLLW regression to Japan's data, Figure \ref{res}a gives
the plot of the modified deviance residuals (\ref{res}) versus the observation
index, whereas Figure \ref{res}b gives the normal probability plot with generated envelope (Atkinson, 1987). Both figures support the LEOLLW regression for modelling these data.\\

\begin{figure}[!htb]
\begin{center}
(a)\hspace{5cm} (b) \\
\includegraphics[height=5cm]{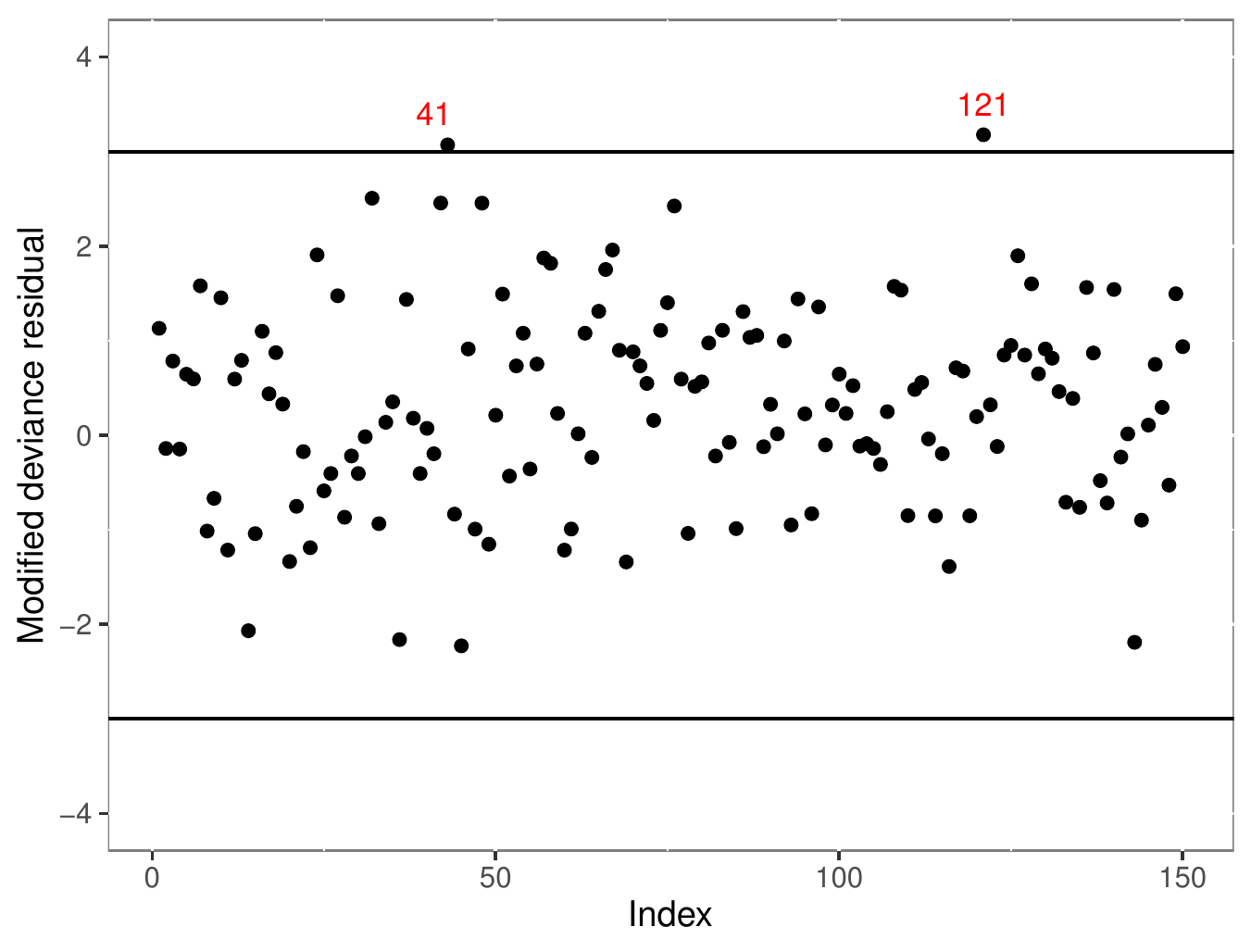}
\includegraphics[height=5cm]{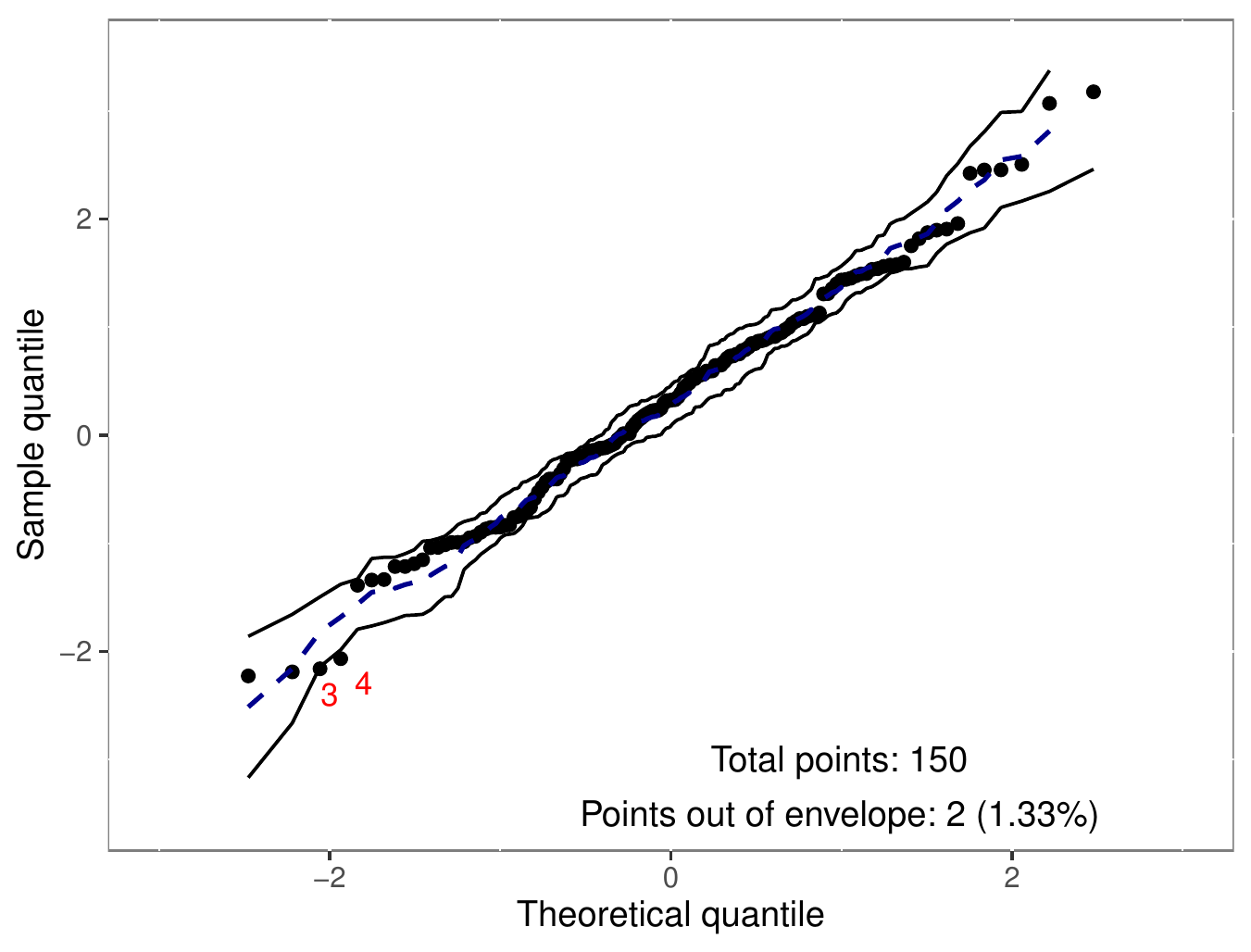}
\end{center}
\caption{(a) Modified deviance residuals versus observation index. (b)
Normal probability plot with envelope for the modified deviance residuals.}
\label{res}
\end{figure}

{\bf Interpretations for $\mu_i$:}\\

\begin{itemize}
  \item There is a significant difference between men and women in relation to the length of stay in Japan.
  \item The length of stay in Japan tends to decrease when the age increases.
  \item A significant difference exists between those who accompanied their family
  and those seeking better living conditions in relation to the length of stay in Japan.
  \item A significant difference exists between those who accompanied their family and those who came for study in relation to the length of stay.
\end{itemize}

{\bf Interpretations for $\sigma_i$:}\\

\begin{itemize}
  \item The variability of the length of stay in Japan tends to decline significantly
   when the age increases.
  \item There is a significant difference between those who accompanied their family and those seeking better living conditions in relation to the variability of the length of stay in Japan.
  \item A significant difference exists between those who accompanied their family and those who came for study in relation to the variability of the length of stay in Japan.
  \item A significant difference exists between those who accompanied their family and those seeking new experiences in relation to the variability of the length of stay.
\end{itemize}

\section{Conclusions}\label{sec:conclusao}

We obtained new mathematical properties of the exponentiated odd log-logistic (EOLL-G) family of distributions. Two new distributions, called the exponentiated odd log-logistic Weibull (EOLLW) and log exponentiated odd log-logistic Weibull (LEOLLW), were proposed and their structural properties were studied.
We defined a new location-scale regression model based on the LEOLLW distribution for censored data, and calculated the maximum likelihood estimates. Some simulations
showed that the empirical distribution of the residuals can be close to the standard normal distribution. We showed that the proposed regression model fitted well to a Japanese-Brazilian emigration data set.

\section*{Acknowledgments}
The authors are very grateful to the editor and two referees for helpful comments.
The financial support from CAPES and CNPq is gratefully acknowledged

\section*{References}
\begin{description}

\item Atkinson, A. C. (1987). {\it Plots, transformations and regression: an introduction to graphical methods of diagnostics regression analysis}.
2nd ed. Oxford: Clarendon Press, 282p.

\item Alizadeh, M., Tahmasebi, S. and Haghbin, H. (2020). The exponentiated odd log-logistic family of distributions: Properties and applications. {\it Journal of Statistical Modelling: Theory and Applications}, {\bf 1}, 29-52.

\item Babbie, E. (1991). {\it M\'etodos de pesquisas de Survey/Earl Babbie}. Coleção Aprender, tradu{\c{c}}{\~o}o de: Survey research methods. Belo Horizonte, Brazil: Editora UFMG.  519p.

\item Burr, I.W. (1942). Cumulative frequency functions. {\it Annals of Mathematical Statistics}, {\bf 13}, 215-232.

\item Collett, D. (2003) {\it Modelling survival data in medical research.} London: Chapman \& Hall. 389p.

\item Cook,  R,D. and Weisberg, S. (1982). {\it Residuals and influence in regression.} New York: Chapman \& Hall. 230p.

\item Cox, D.R. and Snell, E.J. (1968).  A general definition of residuals. {\it Journal of the Royal Statistical Society: Series B}, {\bf 30}, 248–275.

\item Dagum, C. (1975). A model of income distribution and the conditions of existence of moments of finite order. {\it Bulletin of the International Statistical Institute},
 {\bf 46}, 199-205.

\item Fleming, T.T. and Harrington, D.P. (1991). {\it Counting Process and Survival Analysis}. Wiley, New York.

\item Glaser, R.E. (1980). Bathtub and related failure rate characterizations. {\it Journal of the American Statistical
    Association}, {\bf 75}, 667-672.

\item Gleaton, J.U. and  Lynch, J.D. (2006)
Properties of generalized log-logistic families of lifetime
distributions.
{\it Journal of Probability and Statistical Science}, {\bf 4}, 51-64.

\item Griffiths, L. (1947). {\it Introduction to the Theory of Equations}. J. Wiley, New York.

\item Hashimoto, E.M., Ortega, E.M.M., Cordeiro, G.M., Cancho, V.G. and Silva, I. (2021).
The re-parameterized inverse Gaussian regression to model length of stay of COVID-19 patients in the public health care
system of Piracicaba, Brazil.
{\it Journal of Applied Statistics}, DOI:$10.1080/02664763.2022.2036707$.

\item Kawamura, L.K. (1999). Para onde v{\~a}o os brasileiros? Imigrantes brasileiros no Jap{\~a}o.
Campinas, Brazil : Editora da Unicamp, 236p.

\item Mudholkar, G.S., Srivastava, D.K. and Kollia, G. (1996) A generalization of the Weibull distribution with application to the analysis of survival data.
{\it Journal of the American Statistical     Association}, {\bf 91}, 1575–1583

\item Ortega, E.M.M., Paula, G.A. and Bolfarine, H. (2008). Deviance residuals in generalized log-gamma regression models with
censored observations. {\it Journal of Statistical Computation and Simulation}, {\bf 78}, 47–764.

\item R Development Core Team (2022)
{\it R: A Language and Environment for Statistical Computing.}
R Foundation for Statistical Computing, Vienna, Austria.

\item Silva, G.O., Ortega, E.M.M. and Paula, G.A. (2011).
Residuals for log-Burr XII regression models in survival analysis.
{\it Journal of Applied Statistics}, {\bf 38}, 1435-1445.

\item Therneau, T.M., Grambsch, P.M. nd Fleming, T.R. (1990). Martingale-based residuals for survival models. {\it Biometrika}, {\bf 77}, 147–160.

\item Xue, J. (2012). {\it Loop Tiling for Parallelism}. The Springer International Series in Engineering and Computer Science. Springer: US. 256p.


\end{description}

\end{document}